%% file: main.tex
\documentclass[twoside,leqno,twocolumn]{article}
\pdfoutput=1
\usepackage[letterpaper]{geometry}
\usepackage{ltexpprt}

\usepackage{amsmath,amssymb,amsfonts}
\usepackage{algorithm}
\usepackage[noend]{algpseudocode}
\usepackage{booktabs}
\usepackage{mathtools}
\usepackage{cite}
\usepackage[T1]{fontenc}
\usepackage[autolanguage]{numprint}
\usepackage{placeins}
\usepackage[dvipsnames]{xcolor}
\usepackage{xspace}
\usepackage{subcaption}
\usepackage{comment}
\usepackage{url}

\usepackage{dsfont, microtype, xcolor, paralist, environ}

\usepackage[ocgcolorlinks]{hyperref}
\colorlet{DarkRed}{red!50!black}
\colorlet{DarkGreen}{green!50!black}
\colorlet{DarkBlue}{blue!50!black}
\hypersetup{
	linkcolor = DarkRed,
	citecolor = DarkGreen,
	urlcolor = DarkBlue,
	bookmarksnumbered = true,
	linktocpage = true
}

\allowdisplaybreaks[1]

\input{./macros.tex}
\input{./gen_macros_closeness.tex}
\input{./gen_macros_harmonic.tex}

\newcommand{\ArxivOrCr}[2]{#1}

\newcommand*\samethanks[1][\value{footnote}]{\footnotemark[#1]}

\title{Group-Harmonic and Group-Closeness Maximization -- \\ Approximation and Engineering\thanks{
This work has been partially supported by German Research Foundation (DFG) grant ME 3619/3-2 within 
Priority Programme 1736 and 
by the Italian MIUR PRIN 2017 Project ALGADIMAR ``Algorithms, Games, and Digital Markets''.
}}

\author{Eugenio Angriman\thanks{Humboldt-Universit\"at zu Berlin, Department of Computer Science, Unter den Linden 6, 10099 Berlin, Germany.}
\and Ruben Becker\thanks{Gran Sasso Science Institute, L'Aquila, Italy.}
\and Gianlorenzo D'Angelo\samethanks[3]
\and Hugo Gilbert\thanks{Université Paris-Dauphine, Université PSL, CNRS, LAMSADE, 75016 Paris, France.}
\and Alexander van der Grinten\samethanks[2]
\and Henning Meyerhenke\samethanks[2]}

\date{}

\begin{document}

\maketitle

\ArxivOrCr{}{
\fancyfoot[R]{\scriptsize{Copyright \textcopyright\ 2021 by SIAM\\
Unauthorized reproduction of this article is prohibited}}
}

\begin{abstract}
    \ArxivOrCr{\input{abstract_arxiv}}{\input{abstract_cr}}
\end{abstract}

\input{1_introduction}
\input{2_harmonic}
\input{3_closeness}
\input{4_algorithm_engineering}
\input{5_experiments}

\input{6_conclusion}

\bibliography{./references,biblio2}
\newpage 

\ArxivOrCr{\input{appendix}}{}
\end{document}

%% file: macros.tex
\newcommand{\RB}[1]{{\small\textcolor{blue}{RB: #1}}}
\newcommand{\wrt}{w.r.t.\xspace}
\newcommand{\ie}{i.e.\xspace}
\newcommand{\iec}{\ie,\xspace}
\newcommand{\eg}{e.g.,\xspace}
\newcommand{\etal}{et al.\xspace}
\newcommand{\Suv}{\ensuremath{(S\cup \{v\}) \setminus\{u\}}}
\newcommand{\st}{\text{ s.t. }}
\newcommand{\quot}[1]{``#1''}

\newcommand{\false}{\texttt{false}}
\newcommand{\true}{\texttt{true}}
\newcommand{\nil}{\texttt{null}}

\DeclareMathOperator{\diam}{diam}

\DeclareMathOperator{\dist}{d}

\DeclareMathOperator{\argmin}{argmin}
\DeclareMathOperator{\argmax}{argmax}

\DeclareMathOperator{\degout}{\deg_{\text{out}}}

\newcommand{\R}[2]{R(#1,#2)}

\let\epsilon\varepsilon
\let\eps\varepsilon

\newcommand{\NN}{\ensuremath{\mathbb{N}}}

\newcommand{\QQ}{\ensuremath{\mathbb{Q}}}

\newcommand{\Oh}{\mathcal{O}}

\usepackage{tikz}
\usetikzlibrary{decorations.pathreplacing}
\usetikzlibrary{plotmarks}
\usetikzlibrary{positioning,automata,arrows}
\usetikzlibrary{shapes.geometric}
\usetikzlibrary{decorations.markings}
\usetikzlibrary{positioning, shapes, arrows}

\PassOptionsToPackage{usenames,dvipsnames,svgnames}{xcolor}
\definecolor{orange}{RGB}{235,90,0}
\definecolor{darkorange}{RGB}{175,30,0}
\definecolor{turkis}{RGB}{131,182,182}
\definecolor{darkturkis}{RGB}{31,82,82}
\definecolor{green}{RGB}{102,180,0}
\definecolor{darkgreen}{RGB}{51,90,0}
\definecolor{myblue}{RGB}{0,0,213}
\definecolor{mydarkblue}{RGB}{0,0,100}
\definecolor{mybrightblue}{HTML}{74B0E4}
\definecolor{mybrighterblue}{HTML}{B3EAFA}
\definecolor{lila}{RGB}{102,0,102}
\definecolor{darkred}{RGB}{139,0,0}
\definecolor{darkyellow}{RGB}{188,135,2}
\definecolor{brightgray}{RGB}{200,200,200}
\definecolor{darkgray}{RGB}{50,50,50}
\definecolor{amaranth}{rgb}{0.9, 0.17, 0.31}
\definecolor{alizarin}{rgb}{0.82, 0.1, 0.26}
\definecolor{amber}{rgb}{1.0, 0.75, 0.0}
\definecolor{green(ryb)}{rgb}{0.4, 0.69, 0.2}
\definecolor{hanblue}{rgb}{0.27, 0.42, 0.81}
\definecolor{grannysmithapple}{rgb}{0.66, 0.89, 0.63}

\newtheorem{observation}{Observation}

\newproof{@proofof}{Proof of Theorem~\ref{theorem:hardness directed}}
\newenvironment{proofof}{\begin{@proofof}}{\end{@proofof}}

\NewEnviron{cproblem}[1]{%
\begin{center}\fbox{\parbox{.97\columnwidth}{%
    {\centering\scshape #1\par}%
    \parskip=1ex
    \everypar{\hangindent=1em}%
    \BODY
}}
\end{center}
}

\DeclareMathOperator{\VC}{C}
\DeclareMathOperator{\VH}{H}
\DeclareMathOperator{\GC}{GC}
\DeclareMathOperator{\GF}{GF}
\DeclareMathOperator{\GFapx}{\widetilde{GF}}
\DeclareMathOperator{\GFlb}{\widehat{GF}}
\DeclareMathOperator{\GHhat}{\widehat{GH}}
\DeclareMathOperator{\GH}{GH}
\DeclareMathOperator{\LMIN}{\ell_{\min}}
\DeclareMathOperator{\LMAX}{\ell_{\max}}

\newcommand{\tool}[1]{\textsf{#1}\xspace}
\newcommand{\gsls}{\tool{GS-LS-C}}
\newcommand{\grls}{\tool{Greedy-LS-C}}
\newcommand{\greedy}{\tool{Greedy-C}}
\newcommand{\gs}{\tool{GS}}
\newcommand{\rnd}{\tool{Best-Random-C}}
\newcommand{\greedyh}{\tool{Greedy-H}}
\newcommand{\rndh}{\tool{Best-Random-H}}
\newcommand{\grlsh}{\tool{Greedy-LS-H}}

\newcommand{\tabtitle}[1]{
\begin{tabular}{c}
\centering
#1
\end{tabular}}

%% file: gen_macros_closeness.tex
\newcommand{\minQualLSGRCplxUnw}{$99.77\%$\xspace}
\newcommand{\minQualLSGSCplxUnw}{$99.76\%$\xspace}
\newcommand{\minQualLSGRRoadUnw}{$98.66\%$\xspace}
\newcommand{\minQualLSGSRoadUnw}{$98.50\%$\xspace}
\newcommand{\avgSpeedGreedyCplxUnw}{$16.4\times$\xspace}
\newcommand{\minSpeedGSLSCplxUnw}{$28.33\times$\xspace}
\newcommand{\minSpeedGRLSCplxUnw}{$22.99\times$\xspace}
\newcommand{\maxSpeedGSLSCplxUnw}{$233.01\times$\xspace}
\newcommand{\maxSpeedGRLSCplxUnw}{$485.09\times$\xspace}

%% file: gen_macros_harmonic.tex
\newcommand{\minQualLSGRHCplxUnw}{$99.72\%$\xspace}
\newcommand{\minQualGRHRoadUnw}{$98.76\%$\xspace}
\newcommand{\minQualLSGRHRoadUnw}{$99.75\%$\xspace}
\newcommand{\minQualGRHRndCplxDir}{$1.407$\xspace}
\newcommand{\maxQualGRHRndCplxDir}{$1.525$\xspace}
\newcommand{\minQualGRHRndCplxUnd}{$1.445$\xspace}
\newcommand{\maxQualGRHRndCplxUnd}{$1.504$\xspace}
\newcommand{\maxQualImprLSGRHCplxDir}{$0.05\%$\xspace}
\newcommand{\minSlowdLSGRHCplxDir}{$5.7\times$\xspace}
\newcommand{\maxSlowdLSGRHCplxDir}{$27.3\times$\xspace}
\newcommand{\minQualGRHRndRoadDirWei}{$2.4$\xspace}
\newcommand{\maxQualGRHRndRoadDirWei}{$2.6$\xspace}
\newcommand{\maxQualImprLSGRHRoadDirWei}{$0.01\%$\xspace}
\newcommand{\minSlowdLSGRHRoadDirWei}{$54.9\times$\xspace}
\newcommand{\maxSlowdLSGRHRoadDirWei}{$448.9\times$\xspace}
\newcommand{\minSlowdGRHRoadDirWei}{$2.5\times$\xspace}
\newcommand{\maxSlowdGRHRoadDirWei}{$3.6\times$\xspace}
\newcommand{\minQualImprLSGRHRoadUnw}{$0.58\%$\xspace}
\newcommand{\maxQualImprLSGRHRoadUnw}{$0.69\%$\xspace}
\newcommand{\minSlowdGRHRoadUnw}{$3.2\times$\xspace}
\newcommand{\maxSlowdGRHRoadUnw}{$12.2\times$\xspace}

%% file: abstract_arxiv.tex
Centrality measures characterize important nodes in networks. Efficiently computing such nodes has received a lot of attention. When considering the generalization of computing central \emph{groups of nodes}, challenging optimization problems occur. In this work, we study two such problems, \emph{group-harmonic maximization} and \emph{group-closeness maximization} both from a theoretical and from an algorithm engineering perspective.

On the theoretical side, we obtain the following results. For \emph{group-harmonic maximization}, unless $P=NP$, there is no polynomial-time algorithm achieving an approximation factor better than $1-1/e$ (directed) and $1-1/(4e)$ (undirected), even for unweighted graphs. On the positive side, we show that a greedy algorithm achieves an approximation factor of $\lambda(1-2/e)$ (directed) and $\lambda(1-1/e)/2$ (undirected), where $\lambda$ is the ratio of minimal and maximal edge weights. For \emph{group-closeness maximization}, the undirected case is $NP$-hard to be approximated to within a factor better than $1-1/(e+1)$ and a constant approximation factor is achieved by a local-search algorithm. For the directed case, however, we show that, for any $\epsilon<1/2$, the problem is $NP$-hard to be approximated within a factor of $4|V|^{-\epsilon}$.

From the algorithm engineering perspective, we provide efficient implementations of the above greedy and local search algorithms. In our experimental study we show that, on small instances where an optimum solution can be computed in reasonable time, the quality of both the greedy and the local search algorithms come very close to the optimum. On larger instances, our local search algorithms yield results with superior quality compared to existing greedy and local search solutions, at the cost of additional running time. We thus advocate local search for scenarios where solution quality is of highest concern.

%% file: abstract_cr.tex
Centrality measures characterize important nodes in networks. Efficiently computing such nodes has received a lot of attention. When considering the generalization of computing central \emph{groups of nodes}, challenging optimization problems occur.
In this work, we study two such problems, \emph{group-harmonic maximization} and \emph{group-closeness maximization} both from a theoretical and from an algorithm engineering perspective. 

On the theoretical side, we obtain the following results. For \emph{group-harmonic maximization}, unless $P = NP$, there is no polynomial-time algorithm that achieves an approximation factor better than $1-\frac{1}{e}$ (directed) and $1-\frac{1}{4e}$ (undirected), even for unweighted graphs. On the positive side, we show that a greedy algorithm achieves an approximation factor of $\lambda (1-\frac{2}{e})$ (directed) and $\frac{\lambda}{2}(1-\frac{1}{e})$ (undirected), where $\lambda$ is the ratio of minimal and maximal edge weights.
For \emph{group-closeness maximization}, we obtain a strong separation between undirected and directed graphs (that holds even in the unweighted case). The undirected case is $NP$-hard to be approximated to within a factor better than $1-\frac{1}{e+1}$ and a constant approximation factor is achieved by a local-search algorithm. For the directed case, however, we show that, for any $\varepsilon <\frac{1}{2}$, the problem is $NP$-hard to be approximated within a factor of $4|V|^{-\varepsilon}$. 

From the algorithm engineering perspective, we provide efficient implementations of the above greedy and local search algorithms.
In our extensive experimental study we show that, on instances small enough so that an optimum solution can be computed in reasonable time, the quality of both the greedy and the local search algorithms come very close to the optimum. On larger instances, our local search algorithms yield results with superior quality compared to existing greedy and local search solutions, at the cost of additional running time. 
We thus advocate local search for scenarios where solution quality is of highest concern.

%% file: 1_introduction.tex
\section{Introduction}
\label{sec:intro}
%
The identification of important vertices in a graph $G=(V,E)$ is one of the most widely used analytics in
network analysis. To this end, numerous centrality measures have been proposed that
reflect different underlying network processes, see~\cite{boldi2014axioms,ScalingupnetworkcentralitycomputationsAbriefoverview}. Among the widely used ones are
closeness and harmonic centrality, which are both based on shortest-path distances, see~\cite{newman2018networks}.
The textbook algorithm for computing a node ranking \wrt closeness or harmonic centrality 
solves $|V|$ single-source shortest path problems. Top-$k$ ranking queries can often be solved 
faster by suitable pruning~\cite{bergamini2019computing}. Still, closeness is known to be expensive in the worst case: 
one cannot compute the most closeness-central vertex in time $\Oh(|E|^{2-\epsilon})$
for any $\epsilon > 0$ (assuming SETH)~\cite{bergamini2019computing}.

Many network analysis applications do not only require a centrality ranking, but also ask for a \emph{group} of $k$ nodes that is central \emph{as a group}. This is an orthogonal problem since the nodes of the
most central group need to cover the graph together and often differ significantly from the $k$ most highly ranked vertices.
Group centrality problems arise in facility location,
leader selection~\cite{DBLP:journals/tac/ClarkBP14}, and
influence maximization~\cite{banerjee2020survey}, among many others.

Group-closeness and group-harmonic maximization are ${NP}$-hard problems (for group-closeness see~\cite{ChenWW16}).
In practice real-world instances of non-trivial size (say, beyond a few thousand nodes/edges) usually take too long with exact methods such
as ILP solvers~\cite{BergaminiGM18}. Thus, for group-closeness maximization, recent work has concentrated on heuristics~\cite{ChenWW16,BergaminiGM18,AngrimanGM19}.
Early attempts to attribute a constant-factor approximation to a popular greedy algorithm for group-closeness maximization
were flawed (cf.\ discussion in Section~\ref{sub:prelim-discussion}), leaving the question open how and how well both problems can be approximated.

\subsection{Related Work.}
\label{sub:related-work}
Borgatti and Everett~\cite{everett1999centrality} were the first to extend the notion of centrality to groups of nodes,
including group-degree, group-betweenness, and group-closeness.
Group-degree maximization can be reduced from vertex cover (see \eg~\cite{DBLP:conf/alenex/AngrimanGBZGM20}) and is thus ${NP}$-hard. A similar reduction also shows that the optimization of GED-Walk, 
a recent group centrality measure inspired by Katz centrality, is ${NP}$-hard. In the same paper,
Angriman \etal~\cite{DBLP:conf/alenex/AngrimanGBZGM20} prove submodularity of GED-Walk and use a 
well-known greedy strategy to get a constant-factor approximation.
For group betweenness maximization, sampling-based approximation algorithms have been proposed~\cite{mahmoody2016scalable,yoshida2014gbc}.
The notion of group-closeness maximization in~\cite{JunzhouLTG14measuring} differs from the original definition (we use the latter, which is widely accepted)
and only serves as an estimate of the original. That is why their proofs of ${NP}$-hardness
and submodularity do not necessarily carry over to the original one. 
In fact, standard group-closeness maximization
is not sub\-modular (cf.\ Section~\ref{sub:prelim-discussion}), so that
we cannot simply apply submodular optimization results~\cite{vondrak2013symmetry} in this case directly.

Chen \etal~\cite{ChenWW16} argued that group-closeness maximization is ${NP}$-hard by relating it
to the ${NP}$-hard $k$-means problem. Their realization of the common greedy algorithm was later improved by Bergamini \etal~\cite{BergaminiGM18},
who made the algorithm more memory-efficient and exploited the supermodularity of the reciprocal of group-closeness for search pruning.
Exploiting supermodularity of the reciprocal also works when the graph distance is replaced by the 
resistance distance, which leads to the so-called group \emph{current-flow} closeness 
-- for which Li \etal~\cite{0002PSYZ19} proposed approximation algorithms based on greedy strategies and random projections.

Still, even for group-closeness with the usual graph distance, the greedy algorithm can be time-consuming on large instances, which motivated new
local search heuristics~\cite{AngrimanGM19}.
Depending on the quality level expected and the implementation,
these heuristics can be significantly faster than the greedy method
and often obtain a nearly comparable quality. All these works on group-closeness maximization
did not provide approximation bounds (also see Section~\ref{sub:prelim-discussion}), leaving the question of approximability unsettled for both problems considered here.
Yet, the close relationship between group-closeness maximization and the metric $k$-median problem as well as known local search algorithms with constant-factor approximation
bounds for the latter~\cite{AryaGKMMP04} motivate us to investigate
whether the $k$-median results can be transferred to group-closeness and to group-harmonic maximization.

\subsection{Outline and Contribution.}
In this paper, we address theoretical and practical approximation aspects of the two group centrality maximization problems.

In Section~\ref{sec:GH}, we provide the first non-trivial approximation bounds for group-harmonic maximization.
By proving that the problem is submodular, but not necessarily monotone, 
we can directly apply a local search algorithm~\cite{LMNS10} with approximation factor $\frac{1+\varepsilon}{6}$.
We also prove that a greedy algorithm admits a 
$\lambda(1-\frac{2}{e})$-approximation in directed graphs, where
$\lambda$ is the ratio of the smallest and the largest edge weight, respectively. In undirected graphs,
the approximation factor improves to $\frac{\lambda}{2}(1-\frac{1}{e})$. These results have
to be seen in relation to our hardness of approximation results: we show that, unless ${P} = {NP}$,
there is no polynomial-time
algorithm with approximation factor better than $1-1/e$ (general) or $1-1/(4e)$ (undirected).
We proceed by studying group-closeness maximization in Section~\ref{sec:GC}. Interestingly, for this problem we obtain a strong separation between undirected and directed graphs: we prove that the undirected case admits a constant-factor approximation (by relating it to known results on $k$-median).
For the directed case, in turn, we provide the first inapproximability results: it is ${NP}$-hard to approximate the 
problem to within a factor better than $\Theta(|V|^{-\eps})$ for any $\eps<1/2$. 
All our hardness results hold even in the unweighted case, hence the strong separation even holds in the unweighted case. We summarize our results on approximation in Table~\ref{tbl:bounds}.

The purpose of Section~\ref{sec:algo-eng} is to illustrate how to implement greedy and local search heuristics 
(that satisfy approximation guarantees, unlike previous implementations) efficiently for the respective group centrality maximization problem. 
Section~\ref{sec:exp} presents the results
of our experimental study with exact and random restart results as baselines:
where we can make such a comparison, greedy and local search are on average
less than $0.5\%$ away from the optimum.
Local search is one to three orders of magnitude slower than greedy, but this is to be expected due to
a high quality demand; indeed, unlike previous work on local search~\cite{AngrimanGM19} by a subset of the authors, 
our new algorithms often cut greedy's (empirical) gap to optimality by half or more.

\ArxivOrCr{Some proofs are deferred to the appendix along with further experimental results.}
{We refer the reader to a complete version of this article~\cite{abdggm20arxiv} which includes an appendix consisting of some omitted proofs and further experimental results.}

\begin{table}[t]
\centering
{
\resizebox{\columnwidth}{!}{\begin{tabular}{l|c|c|c|c|c|c}
  &\multicolumn{3}{c|}{Directed graphs}&\multicolumn{3}{c}{Undirected graphs}\\
\hline
  &Greedy&LS&Hardn.&Greedy&LS&Hardn.\\
\hline
$\GH$&$\lambda\left(1-\frac{1}{2e}\right)$&$\frac{1+\epsilon}{6}$&$1-\frac{1}{e}$&$\frac{\lambda}{2}\left(1-\frac{1}{e}\right)$&$\frac{1+\epsilon}{6}$&$1-\frac{1}{4e}$\\
\hline
$\GC$&&&$4|V|^{-\epsilon}$&&$3+\frac{2}{p}$&$1-\frac{1}{e+1}$\\
\end{tabular}
}}
\caption{A summary of our approximation bounds: $\lambda$ is the ratio of the minimum and maximum edge weight, respectively, $\epsilon$ is any value in $(0,1/2)$, $p$ is the number of swaps in the local search. The hardness of approximation bounds hold even in the unweighted case.}
\label{tbl:bounds}
\end{table}

\section{Preliminaries}
In all the problems we study, we are given a weighted (possibly directed) graph $G=(V, E, \ell)$, where $|V| = n$ and $\ell:E\rightarrow \NN_{>0}$ is an edge-weight function. We do not  assume that $G$ is connected,
but we assume that there are no isolated nodes. For two nodes $u,v\in V$, we denote with $\dist(u,v)$ the length of a shortest path in $G$ from $u$ to $v$, where length is measured \wrt the function $\ell$. We denote by $\LMIN:= \min_{e\in E} \ell(e)$ and $\LMAX:= \max_{e\in E} \ell(e)$ the lowest and highest  weights in graph $G$, respectively. We furthermore let 
$\lambda :=\frac{\LMIN}{\LMAX}$ be the ratio of smallest and largest edge weight.

\paragraph{Centrality Measures.} 
To measure the relative importance of a vertex in a graph, different centrality measures have been defined. 
Notably, two well-known centrality measures based on distances are \emph{closeness centrality} and \emph{harmonic centrality}. 
Formally, the closeness centrality $\VC(u)$ and harmonic centrality $\VH(u)$ of a vertex $u$ are defined as follows:
\begin{align*}
    &\VC(u):=\frac{n}{\sum_{v\in V\setminus \{u\}} \dist(u, v)}\text{ and } \quad\\
    &\VH(u):= \sum_{v\in V\setminus \{u\}} \frac{1}{\dist(u, v)}.
\end{align*}

These two centrality measures differ by the order in which they apply the sum and inverse operations. 
Note that while closeness centrality may seem more natural than harmonic centrality, it suffers from its inability to address disconnected graphs.  
Indeed, note that in this case, all vertices have a centrality of zero. 
This finding has been one of the motivations to introduce harmonic centrality, which 
additionally 
enjoys several desirable properties from an axiomatic viewpoint~\cite{boldi2014axioms}.

In this work, we study the generalizations of these two centrality measures to groups of nodes. 
We start by extending the notion of distances to sets by defining $\dist(S, v):=\min_{u\in S}\dist(u,v)$. 
In words, $\dist(S, v)$ denotes the distance from the closest node in $S$ to $v$. 
This notation allows us to define the group-closeness and group-harmonic centrality measures. 

\paragraph{Group Centralities.} 
For a group $S\subset V$ of vertices in $G$, its \emph{group-closeness centrality} is defined as
\[
    \GC(S):=\frac{n}{\sum_{v\in V\setminus S} \dist(S, v)},
\]
see for example~\cite{BergaminiGM18}. A similar objective has been addressed in the literature as well, namely the farness of a set, defined by $\GF(S):=\frac{1}{n}\cdot \sum_{v\in V\setminus S} \dist(S, v)$.
We note that the farness of a set is the reciprocal of its closeness.

The \emph{group-harmonic centrality} of a group $S\subset V$ of vertices in $G$ is defined as 

\[
\GH(S) := \sum_{v\in V\setminus S} \frac{1}{\dist(S,v)},
\]
where $\frac{1}{\dist(S,v)}=0$ if there is no path from $S$ to $v$. 
While this definition provides a natural generalization to harmonic centrality, the way it handles the vertices in the set $S$ may seem questionable. 
Indeed, why should these vertices
count as 0 while they are the closest ones to the group? 
On the other hand, making them count more than 0 by assigning them an arbitrary value would also be unsatisfactory. 
We work around this problem by always comparing the group-harmonic centrality of sets of equal cardinality. 
Indeed,
the value assigned to vertices in the set does not impact such comparisons.

\paragraph{Computational Problems.} 
In this work, we study the following two computational problems that consist of finding groups that maximize the two introduced centrality measures with respect to a budget constraint, \iec for a given parameter $k$, we are interested in finding a group of size $k$ of large group-closeness centrality or group-harmonic centrality.
Formally:
\begin{cproblem}{group-closeness maximization}
    Input: Graph \(G=(V, E, \ell)\), integer \(k\).

    Find: Set \(S\subset V\) with \(|S| = k\), s.t.\ \(\GC(S)\) is maximum.
\end{cproblem}

\begin{cproblem}{group-harmonic maximization}
    Input: Graph \(G=(V, E, \ell)\), integer \(k\).

    Find: Set \(S\subset V\) with \(|S| = k\), s.t.\ \(\GH(S)\) is maximum.
\end{cproblem}

While group-closeness maximization has already been studied in several settings~\cite{AngrimanGM19}, to the best of our knowledge, we are the first to study the group-harmonic maximization problem. 

%% file: 2_harmonic.tex
\section{Group-Harmonic Maximization} \label{sec:GH}
\subsection{Mathematical Properties.}
We start our study of the group-harmonic maximization problem by analyzing 
the mathematical properties of the set function $\GH(\cdot)$. We observe that, while the function is submodular, it is not monotone.

\begin{lemma}\label{harmonic:submodularity}
Function $\GH:2^V\rightarrow \QQ_{\ge 0}$ is submodular.
\end{lemma}

To see
that the function $\GH(\cdot)$ is not necessarily monotone,
consider the example of an undirected graph with two nodes $u,v$ and one edge between them, then $\GH(\{u\}) = \GH(\{v\}) = 1$, while $\GH(\{u,v\})=0$.

\subsection{Approximation Algorithms.}
As $\GH(\cdot)$ is submodular, we can use the local-search algorithm due to Lee et al.~\cite{LMNS10} and obtain a $\left(\frac{1+ \epsilon}{6}\right)$-approximation (the exact cardinality constraint corresponds to the case of a single matroid base constraint, where the matroid is the uniform one). 
This algorithm was notably improved by Vondr\'ak~\cite{vondrak2013symmetry}, who designed a randomized local-search method with an approximation factor of $\left(\frac{1}{4} - o(1)\right)$. 
Another approximation algorithm candidate is the  greedy algorithm (Algorithm~\ref{algo:mon+submodular:greedy}) that provides an approximation factor of $1-\frac{1}{e}$ for maximizing a monotone and submodular function under a cardinality constraint $|S|\le k$. 

\begin{algorithm}
\caption{Greedy algorithm for maximizing a monotone submodular function $f$}
\label{algo:mon+submodular:greedy}
\begin{algorithmic}[1]
\State $S \gets \emptyset$
\While{$|S| < k$}
  \State $v \gets \argmax_{u\in V\setminus S}\{f(S\cup\{u\})-f(S)\}$
  \State $S \gets S\cup \{v\}$
\EndWhile
\State\Return $S$
\end{algorithmic}
\end{algorithm}

Unfortunately, as $\GH(\cdot)$ is not monotone, we cannot use this result directly. 
However, in what follows, we show that Algorithm~\ref{algo:mon+submodular:greedy} still guarantees interesting approximation bounds despite the non-monotonicity of $\GH(\cdot)$. 
Indeed, we obtain the following Theorem.
\begin{theorem}\label{harmonic:theoremapx2}
Algorithm~\ref{algo:mon+submodular:greedy} guarantees the following approximation factors for the group-harmonic maximization problem, where $\lambda :=\frac{\LMIN}{\LMAX}$ is the ratio of the minimum and maximum edge weight.
\begin{itemize}
\item $\lambda(1-\frac{2}{e}) > 0.264\lambda$ in the directed case;
\item $\frac{\lambda}{2}\left(1-\frac{1}{e}\right) > 0.316\lambda$ in the undirected case.
\end{itemize}
\end{theorem}

While these approximation factors may be worse than the ones provided by 
Lee et al.~\cite{LMNS10} and Vondr\'ak~\cite{vondrak2013symmetry}, they offer better guarantees 
for the unweighted version of the group-harmonic maximization problem.

We prove Theorem~\ref{harmonic:theoremapx2} by showing that the corresponding approximation factors hold in the unweighted case and then using the following lemma.

\begin{lemma}\label{harmonic:lemmaweighted}
An $\alpha$-approximation algorithm for the unweighted case of the group-harmonic maximization problem 
yields an $\alpha\lambda$-approximation algorithm for the general case, where $\lambda :=\frac{\LMIN}{\LMAX}$ is the ratio of the minimum and maximum edge weight.
\end{lemma}

We now analyze Algorithm~\ref{algo:mon+submodular:greedy} in the unweighted case. 
Let $S_i$ be the set computed by Algorithm~\ref{algo:mon+submodular:greedy} at the end of iteration $i$ 
and $\Delta_i = \max_{u\in V\setminus S_{i-1}}\{f(S_{i-1}\cup\{u\})-f(S_{i-1})\}$. 
We first show an approximation result in case of $\Delta_i < 0$ (directed) or $\Delta_i \leq 0$ (undirected), respectively.

\begin{lemma} \label{lemma: one half approx}
If $G$ is directed (resp. undirected), and if at some iteration $i\in\{1,\ldots, k\}$, $\Delta_i < 0$ (resp. $\Delta_i \le 0$), then the set returned by Algorithm~\ref{algo:mon+submodular:greedy} provides a $0.5$-approximation.
\end{lemma}

\paragraph{Analysis in the Directed Case.} 
Let us consider a set $\overline{O} \in \argmax \{ \GH(S): |S|\le k \}$ 
with smallest size $k'\leq k$. 
Note that $\overline{O}$ may be of size smaller than $k$,
whereas group-harmonic maximization asks for solutions of size exactly $k$. 
Observe that for each node $v\in \overline{O}$, we have that there exists a node $u\in V \setminus \overline{O}$ whose distance from $\overline{O}$ is due to node $v$, that is $\dist(\overline{O}\setminus \{v\},u) > \dist(\overline{O},u)$, as otherwise we can find an optimal solution with smaller size. 
This implies that $\GH(\overline{O}) \geq k'$.

Let us consider the function $h'(S) := \GH(S) + |S|$. First note that $\overline{O}$ is an optimal solution of size $k'$ for $h'$. 
Secondly, note that $h'$ is monotone, as for each $v\in V\setminus S$ 
$ h'(S\cup\{v\}) \ge \GH(S) - \frac{1}{\dist(S,v)}  + |S| +1\geq h'(S) $, 
and submodular, as it is the sum of two submodular functions. 
Moreover, note that the greedy algorithm shows the same behavior for $h(\cdot)$ and $h'(\cdot)$.
Hence, we obtain that 
\begin{align*}
h'(S_{k'}) &= \GH(S_{k'}) + k' 
 \geq \left(1-\frac{1}{e}\right)h'(\overline{O}) \\
 &= \left(1-\frac{1}{e}\right)(\GH(\overline{O}) +k'),
\end{align*}
where $S_k'$ is the set obtained at iteration $k'$ of Algorithm~\ref{algo:mon+submodular:greedy}. Hence, 
we obtain that
\begin{align*}
\GH(S_{k'}) &\geq \left(1-\frac{1}{e}\right)(\GH(\overline{O}) +k') - k' \\
&= \left(1-\frac{1}{e}\right)\GH(\overline{O}) -\frac{k'}{e}\\
&\geq \left(1-\frac{1}{e}\right)\GH(\overline{O}) -\frac{\GH(\overline{O})}{e} \\
&= \left(1-\frac{2}{e}\right)\GH(\overline{O})\ge \left(1-\frac{2}{e}\right)\GH(O),
\end{align*}
where $O$ is an optimal solution to the group-harmonic maximization problem. 
Let $S$ be the solution returned by Algorithm~\ref{algo:mon+submodular:greedy}. 
If for all iterations $i$ of the algorithm $\Delta_i$ is greater than or equal to 0, then we obtain that $\GH(S) \geq \GH(S_{k'}) \ge \left(1-\frac{2}{e}\right)\GH(O)$. 
Otherwise, by Lemma~\ref{lemma: one half approx}, we obtain that $S$ is a 0.5-approximation, which is better than the claimed approximation bound. \\

\paragraph{Analysis in the Undirected Case.}
Let $O$ be an optimal solution to group-harmonic maximization. We start with the following lemma lower bounding the increment
achieved in each iteration of Algorithm~\ref{algo:mon+submodular:greedy}.

\begin{lemma}\label{harmonic:lemmaincrement}
For each $i=1,\ldots,k$, it holds that $\GH(S_i)-\GH(S_{i-1}) \geq \frac{1}{k}\left( \GH(O) - \GH(S_{i-1}) \right) - 1$.
\end{lemma}
\begin{proof}
For a set of nodes $T$, let us partition the set $V\setminus T$ into sets $\R{u}{T}$ for $u\in T$, 
where $v\in \R{u}{T}$ if $\dist(u,v)=\dist(T,v)$; ties are broken arbitrarily in such a way that 
$\left\{ \R{u}{T} \right\}_{u\in T}$ is a partition of $V\setminus T$. Then:
{\small 
\begin{align}
\GH(O) &- \GH(S_{i-1}) = \!\!\!\sum_{v\in V\setminus O} \!\!\!\frac{1}{\dist(O,v)} \!-\!\!\!\!\! \sum_{v\in V\setminus S_{i-1}} \!\!\!\!\frac{1}{\dist(S_{i-1},v)}\nonumber\\
                 & = \!\!\!\!\!\!\sum_{v\in V\setminus(O\cup S_{i-1})} \!\!\left(  \frac{1}{\dist(O,v)} \!-\! \frac{1}{\dist(S_{i-1},v)} \right)\label{harmonic:lemmaincrement:one}\\ 
                 & \nonumber \qquad + \!\!\! \sum_{v\in S_{i-1}\setminus O} \! \frac{1}{\dist(O,v)} \!-\!\!\! \sum_{v\in O\setminus S_{i-1}} \!\!\frac{1}{\dist(S_{i-1},v)}\\
                 & =  \sum_{u\in O}\sum_{v\in \R{u}{O}\setminus S_{i-1}} \!\!\! \left(  \frac{1}{\dist(u,v)} \!-\! \frac{1}{\dist(S_{i-1},v)} \right) \label{harmonic:lemmaincrement:two} \\ 
                 & \nonumber \qquad + \sum_{v\in S_{i-1}\setminus O} \frac{1}{\dist(O,v)} - \sum_{v\in O\setminus S_{i-1}} \frac{1}{\dist(S_{i-1},v)}\\
                 & \leq \!\!\!   \sum_{u\in O\setminus S_{i-1}}\sum_{v\in \R{u}{O}\setminus S_{i-1}}\!\!\!\! \left(  \frac{1}{\dist(u,v)} \!-\! \frac{1}{\dist(S_{i-1},v)} \right)\label{harmonic:lemmaincrement:three}\\
                 & \nonumber \qquad +\!\! \sum_{v\in S_{i-1}\setminus O} \frac{1}{\dist(O,v)} - \sum_{v\in O\setminus S_{i-1}}\!\! \frac{1}{\dist(S_{i-1},v)}\\
                 & \leq \!\!\!\!\! \sum_{u\in O\setminus S_{i-1}} \!\! \left[ \sum_{v\in \R{u}{O}\setminus S_{i-1}} \!\!\!\!\!\! \left( \frac{1}{\dist(u,v)} - \frac{1}{\dist(S_{i-1},v)}  \right) \right.\label{harmonic:lemmaincrement:four} \\ 
                 & \nonumber \qquad \left. - \frac{1}{\dist(S_{i-1},u)} \right] + k,
\end{align}}
where equality~\eqref{harmonic:lemmaincrement:one} is a reordering of the terms, equality~\eqref{harmonic:lemmaincrement:two} holds since $\left\{ \R{u}{O} \right\}_{u\in O}$ is a partition of $V\setminus O$, inequality~\eqref{harmonic:lemmaincrement:three} holds because, for $u\in O\cap S_{i-1}$ and $v\in \R{u}{O}$, we have $\dist(u,v)\geq \dist(S_{i-1},v)$ and then $\frac{1}{\dist(u,v)} - \frac{1}{\dist(S_{i-1},v)}\leq 0$, and inequality~\eqref{harmonic:lemmaincrement:four} holds because $\frac{1}{\dist(O,v)}\leq 1$
for each $v\in S_{i-1}\setminus O$ and $|S_i|\le k$.

Let $\bar{v}$ be the node selected
at  iteration $i$, i.e., $S_i\setminus S_{i-1} = \{\bar{v}\}$; then, for each $u\in O\setminus S_{i-1}$, we have
{\small \begin{align}
\GH(S_i) &- \GH(S_{i-1}) = \!\!\! \sum_{v\in V\setminus S_i} \! \frac{1}{\dist(S_i,v)} \!- \!\!\! \sum_{v\in V\setminus S_{i-1}} \!\! \frac{1}{\dist(S_{i-1},v)}\nonumber\\
                    &= \!\!\! \sum_{v\in \R{\bar{v}}{S_i}} \!\!\! \left( \frac{1}{\dist(\bar{v},v)} \!-\! \frac{1}{\dist(S_{i-1},v)} \right) \!-\! \frac{1}{\dist(S_{i-1},\bar{v})}\nonumber\\
                    &\geq \!\!\!\!\!\!\!\! \sum_{v\in \R{u}{S_{i-1}\cup\{u\}}} \!\! \left( \frac{1}{\dist(u,v)} - \frac{1}{\dist(S_{i-1},v)} \right) - \frac{1}{\dist(S_{i-1},u)}\label{harmonic:lemmaincrement:seven}\\
                    &\geq \!\!\!\!\!\! \sum_{\substack{v\in \R{u}{S_{i-1}\\\cup\{u\}}\cap \R{u}{O}}} \! \left( \frac{1}{\dist(u,v)}  \!-\! \frac{1}{\dist(S_{i-1},v)} \right) - \frac{1}{\dist(S_{i-1},u)}\label{harmonic:lemmaincrement:eight} \\
                    &\geq \!\!\!\!\!\! \label{harmonic:lemmaincrement:nine}\sum_{\substack{v\in \R{u}{S_{i-1}\\\cup\{u\}}\cap \R{u}{O}}}\left( \frac{1}{\dist(u,v)} - \frac{1}{\dist(S_{i-1},v)} \right) \\
                    &\nonumber \qquad + \!\!\!\! \sum_{\substack{v\in \R{u}{O}\setminus (\R{u}{S_{i-1}\\\cup\{u\}} \cup S_{i-1})} }\!\!\!\left( \frac{1}{\dist(u,v)} - \frac{1}{\dist(S_{i-1},v)} \right)\\
                    & \nonumber \qquad \qquad - \frac{1}{\dist(S_{i-1},u)}\\
                    &=\!\!\!\! \sum_{v\in \R{u}{O}\setminus S_{i-1}} \!\!\!\! \left( \frac{1}{\dist(u,v)} \!-\! \frac{1}{\dist(S_{i-1},v)}  \right) \!-\! \frac{1}{\dist(S_{i-1},u)}\label{harmonic:lemmaincrement:ten},
\end{align}}
where inequality~\eqref{harmonic:lemmaincrement:seven} holds since node $\bar{v}$ is the one that maximizes the marginal increment and $u$ is available at iteration $i$, inequality~\eqref{harmonic:lemmaincrement:eight} follows since for each $v\in \R{u}{S_{i-1}\cup\{u\}}$, we have $\dist(u,v) \leq \dist(S_{i-1},v)$ and then all the terms in the sum are non-negative, while, for nodes $v\in \R{u}{O}\setminus \R{u}{S_{i-1}\cup\{u\}}$, we have $\dist(u,v) \geq \dist(S_{i-1},v)$, since there is a shortest path from $S_{i-1}$ to $v$ that does not pass through $u$, and hence all the terms in the second sum of~\eqref{harmonic:lemmaincrement:nine} are non-positive, which implies the last inequality.
Combining equations (\ref{harmonic:lemmaincrement:four}) and (\ref{harmonic:lemmaincrement:ten}), we have
\begin{align*}
\GH(O) \!-\! \GH(S_{i-1}) &\leq \!\!\!\!\! \sum_{u\in O\setminus S_{i-1}} \!\!\!\!\! \left( \GH(S_i) \!-\! \GH(S_{i-1}) \right) \!+\! k\\
&\leq k\cdot\left( \GH(S_i) - \GH(S_{i-1}) \right) + k,
\end{align*}
since $|O| = k$, which implies the statement.
\hfill\end{proof}

We can now prove that Algorithm~\ref{algo:mon+submodular:greedy} guarantees an approximation factor of $\frac{1}{2}\left(1-\frac{1}{e}\right)$ in the unweighted undirected case.

\begin{proof}
As we saw in Lemma~\ref{lemma: one half approx}, Algorithm~\ref{algo:mon+submodular:greedy} provides a $0.5$-approximation, which is larger than the claimed approximation ratio, if at some iteration $\Delta_i \leq 0$. Hence, we now assume that in all iterations of Algorithm~\ref{algo:mon+submodular:greedy}, we have $\Delta_i > 0$. 
In this case, we prove by induction that
\begin{equation}\label{harmonic:theoremapx2:induction}
  \GH(S_i)\geq \left(1-\left(1-\frac{1}{k}\right)^i\right)\GH(O) - i,
\end{equation}
for each iteration $i=1,\ldots,k$. The inductive basis is implied by Lemma~\ref{harmonic:lemmaincrement}, since for $i=1$ we have $\GH(S_1)\geq \frac{\GH(O)}{k} - 1$.
For $i>1$, by Lemma~\ref{harmonic:lemmaincrement} and the inductive hypothesis, we have
\begin{align*}
\GH(S_i) &= \GH(S_i) -\GH(S_{i-1}) + \GH(S_{i-1})\\
       &\geq \frac{1}{k}\left( \GH(O) - \GH(S_{i-1}) \right) - 1 + \GH(S_{i-1})\\
       &=\GH(S_{i-1})\left( 1-\frac{1}{k} \right) + \frac{1}{k}\GH(O) - 1\\
       &\geq\left[ (1- \left( 1-\frac{1}{k} \right)^{i-1})\GH(O) - i + 1 \right] \! \cdot \! \left(1-\frac{1}{k}\right)\\
       & \quad + \frac{1}{k}\GH(O) - 1\\
       &= \GH(O)\left( 1-\left( 1-\frac{1}{k} \right)^i \right) - i + \frac{i-1}{k}\\
       &\geq \GH(O)\left( 1-\left( 1-\frac{1}{k} \right)^i \right) - i.
\end{align*}
We now show that $\GH(S_i)\geq i$.
This claim is due to the fact that the number of nodes at distance one from $S_i$ is greater than $i$,
which we prove by induction.
The claim is clear for $S_1$. Let us assume the claim to be true at iteration $i-1$ and let $u$ be the node picked by the greedy algorithm at iteration $i$. 
First observe that there exists a distinct neighbor $v$ of $u$ which is at distance at least $2$ from $S_{i-1}$ since we assumed that $\Delta_i > 0$. 
If $d(S_{i-1},u)\ge 2$, then we are done. If $d(S_{i-1},u) =  1$, we prove by contradiction that $|\{v~|~d(u,v) = 1 \text{ and }  d(S_{i-1},v) \ge 2\}|\ge 2$. 
Let this set be a singleton $\{v\}$; we show that picking $v$ would yield a higher increment than $u$, a contradiction.
Indeed, let $T_u = \{ w~|~d(u,w) < d(S_{i-1},w)\}$. Note that $T_u$ contains necessarily other nodes than $v$,
as otherwise $u$ would not yield a positive increment. 
Vertex $v$ is closer than $u$ to all vertices in $T_{u} \setminus \{v\}$. 
Lastly, let $h^{u,v}(S) = \sum_{w\in \{u,v\} \setminus S} 1/d(S,w)$; then $h^{u,v}(S_{i-1}\cup \{u\}) - h^{u,v}(S_{i-1}) = h^{u,v}(S_{i-1}\cup \{v\}) - h^{u,v}(S_{i-1}) = - 1/2$.  
By this observation and Equation~(\ref{harmonic:theoremapx2:induction}) it follows that
\[
2\GH(S_i) \geq i + \GH(S_i) \geq \left(1-\left(1-\frac{1}{k}\right)^i\right)\GH(O).
\]
By setting $i=k$, we get that $\GH(S)$ is at least
\begin{align*}
 \frac{1}{2}\left(1-\left(1-\frac{1}{k}\right)^{k}\right)\GH(O) 
 \geq \frac{1}{2}\left(1-\frac{1}{e}\right)\GH(O),
\end{align*}
which concludes the proof.
\hfill\end{proof}

\subsection{Hardness Results.}
We conclude this section with two hardness of approximation results, one in the directed case and one in the undirected case. 
These results do not completely close the gap w.r.t.\ the guarantees of our approximation algorithms, but they provide 
first upper bounds on the approximation factors that can be achieved for group-harmonic maximization.
\begin{theorem}\label{harmonic:hardness:directed}
Even in the unweighted case, there is no polynomial-time algorithm that can approximate group-harmonic maximization with a factor greater than $1-1/e$, unless ${P}={NP}$.
\end{theorem}
\begin{proof}
We provide a simple reduction from the maximum coverage problem which is known to be hard to approximate better than $1 - 1/e$~\cite{feige1996threshold}. 
In the maximum coverage problem, we are given a universe $U = \{x_1, \ldots, x_n\}$ of $n$ elements, a collection $C = \{S_1,\ldots, S_m\}$ of $m$ subsets of $U$ and a positive integer $k$. 
The goal is to select $k$ sets $\{S_{i_1},\ldots,S_{i_k}\}$ in $C$ that maximize $|\bigcup_{j=1}^{k} S_{i_j}|$. 
Given an instance $(U,C,k)$, we create the following unweighted digraph. 
There exists a vertex $v_x$ for each element $x\in U$ and a vertex $v_S$ for each set $S\in C$. 
Moreover, there is one arc from $v_S$ to $v_x$ if $x\in S$. 
Then we consider the group harmonic maximization instance defined by this digraph and a budget $k$.
The soundness of the reduction stems from the following two observations. (1) To maximize the group harmonic maximization problem, one should only select vertices associated to sets. (2) For a solution $T$ only compounded of vertices associated to sets, $\GH(T) = |\bigcup_{v_{S}\in T} S|$.\hfill
\end{proof}

\begin{theorem}\label{harmonic:hardness:undirected}
Even in the unweighted undirected case, there is no polynomial-time algorithm that can approximate group-harmonic maximization with a factor greater than $1-1/4e$, unless ${P}={NP}$.
\end{theorem}
In order to prove the theorem, we assume that there exists a $\gamma$-approximation algorithm $\mathcal{A}$ for the group harmonic maximization problem, where $\gamma > 1-1/4e$. We then show that, using $\mathcal{A}$, we can get a logarithmic-factor approximation algorithm to the minimum dominating set problem, which is not possible unless $P=NP$~\cite{DinurS14}. \ArxivOrCr{See Appendix~\ref{apx:reduction harmonic undirected} for the proof.}{The complete proof can be found in the full version~\cite{abdggm20arxiv}.}

%% file: 3_closeness.tex
\section{Group-Closeness Maximization}\label{sec:GC}
\subsection{Preliminary Discussion.}
\label{sub:prelim-discussion}
Different variants of the group-closeness maximization problem occur depending on whether the graph at hand is undirected or directed. 
When studying these problems from an approximation algorithm's perspective, it is tempting to observe that the group-farness $\GF(\cdot)$ is a supermodular set function.
In the literature, see the paper by Chen et al.~\cite{ChenWW16}, this has been used to argue that $\GC(\cdot)$ is submodular by falsely assuming that 
the reciprocal of a supermodular function was submodular. It is well-known~\cite{NemhauserWF78} that maximizing a submodular set function with respect to 
a cardinality constraint can be done using the greedy algorithm within an approximation factor of $1-1/e$. Unfortunately, this approach is flawed and thus the 
approximation question can be considered as still unresolved for this problem. 
\ArxivOrCr{In Appendix~\ref{apx:example}, we provide a counter-example to the submodularity of $\GC(\cdot)$.}{In the full version of this article, we provide a counter-example to the submodularity of $\GC(\cdot)$\cite[Appendix~B]{abdggm20arxiv}.}

A similar, yet non-flawed, approach has been recently taken by Li et al.~\cite{0002PSYZ19}. 
In their work, which deals with a different notion of group centrality, namely ``current-flow closeness centrality'', they measure the approximation factor of their algorithms in a different way, allowing them to
obtain constant-factor approximation results. 
\ArxivOrCr{In Appendix~\ref{apx:li-approach}, we argue that an approach similar to theirs can be applied also in our setting, yielding constant-factor approximation algorithms in their sense.}{In the full version of this article~\cite[Appendix~C]{abdggm20arxiv}, we argue that an approach similar to theirs can be applied also in our setting, yielding constant-factor approximation algorithms in their sense.}
We would like to stress, however, that this notion of approximation used in the work by Li et al.\ is a fundamentally different notion of approximation.

\subsection{Approximation Algorithms.}
We will observe that the undirected and directed problems fundamentally differ from an approximation algorithm's perspective when considering the standard notion of approximation factor. 
Indeed, while the undirected case allows for a constant-factor approximation, it is ${NP}$-hard to approximate directed group-closeness maximization to within a factor better than $\Theta(n^{-\eps})$ for any $\eps<1/2$. We stress that this strong separation between the directed and undirected case even occurs in the unweighted case.

We start by introducing the metric $k$-Median problem following Arya et al.~\cite{AryaGKMMP04}.
\begin{cproblem}{Metric $k$-Median}
    Input: Set of clients $C$, set of facilities $F$, cost matrix $c$ with $c_{i,j}\ge 0$ for $i\in C$, $j\in F$, satisfying triangle inequality, integer \(k\).

    Find: Set \(S\subseteq F\) with \(|S|\le k\), s.t.\ \(c(S):=\sum_{i\in C}\min_{j\in S} c_{i,j}\) is minimum.
\end{cproblem}
Arya et al.\ show that the local search algorithm that performs $p$ swaps at a step leads to a solution with approximation ratio at most $3+2/p$ for Metric $k$-Median. 

The group-farness problem can be seen as a special case of the metric $k$-Median problem where $C$ and $F$ are both taken to be the vertex set and the cost matrix being obtained using the shortest path distances.
Since $\GF(\cdot)$ is monotone, the result of Arya et al.\ carries over to the undirected group-farness maximization problem with exact cardinality constraint, yielding an approximation factor of $\frac{p}{3p +2}$ for group-closeness maximization.

\subsection{Hardness Results.}%
\paragraph{The Undirected Case.}
Following~\cite{DAngeloDNP16}, it is ${NP}$-hard to approximate the metric $k$-median problem
to within a factor of $1+1/e$ even in the case when sets $C$ and $F$ are the same set.
This is equivalent to the group-farness minimization problem in undirected connected graphs and hence we get that it is ${NP}$-hard to approximate this latter to within a factor $1+1/e=(e+1)/e\approx 1.37$.
Similarly, we get that it is ${NP}$-hard to approximate the undirected group-closeness maximization problem to within a factor of $e/(e+1)=1-1/(e+1)\approx 0.73$.

\paragraph{The Directed Case.}
In this paragraph, we prove that the directed case fundamentally differs from the undirected case from an approximation algorithm's perspective, that is we will show the following theorem:
\begin{theorem}\label{theorem:hardness directed}
It is ${NP}$-hard to approximate the group-closeness maximization problem within $4\cdot|V|^{-\epsilon}$ for any $\epsilon \in (0,1/2)$, even in the case of an unweighted DAG.
\end{theorem}
To prove this result, we provide a reduction from the Set Cover problem. Let $X = \{x_1,\ldots,x_n\}$ be a universe, $C = \{C_1,\ldots,C_m\}$ be a collection of subsets of $X$ and $k$ be an integer. The Set Cover problem investigates whether there exists a subset $C'\subseteq C$ of size at most $k$ such that $\bigcup_{C_i\in C'} C_i = X$. Importantly, note that the Set Cover problem is NP-hard even if $m$ is less than or equal to $3n$. Indeed, there is a simple reduction from Exact Cover by 3-Sets to the Set Cover problem and Exact Cover by 3-Sets remains NP-hard when each element appears in exactly three subsets~\cite{Gonzalez85}.

\emph{The reduction.} Let $(X,C,k)$ be a Set Cover instance and let $\delta > 0$ be arbitrary. We construct the following instance of the group-closeness centrality problem: For each set $C_j$, we create  $\alpha:=\lceil n^{1+\delta} \rceil$ vertices $\{q_j^\ell : \ell \in [\alpha]\}$ and connect them in the form of a path by $\alpha -1$ arcs $\{(q_j^\ell, q_j^{\ell+1}):\ell \in [\alpha - 1]\}$. For each element $x_i$, we create $\Lambda := m\alpha^2$ vertices $\{p_i^t: t\in[\Lambda]\}$ and arcs $(q_j^{\alpha},p_i^t)$ for all $t\in[\Lambda]$, if $x_i\in C_j$. Lastly, we add a vertex $s$ and arcs $(s,q_j^{1})$ for all $j\in [m]$. The budget is set to $k+1$. The reduction is illustrated in Figure~\ref{figure:reduction directed}. The number of resulting vertices $V$ can be bounded as $|V| = m\alpha + nm\alpha^2 + 1 \le 4n^{4+2\delta}$ using that \(m \le 3n\) and assuming that $n\ge 18$. Indeed,
\begin{align*}
   |V| & =  m\alpha + nm\alpha^2 + 1 \\
       &\le m(n^{1+\delta}+1) + nm(n^{1+\delta}+1)^2 + 1\\
                             &\le mn^{1+\delta}+ m + mn^{3+2\delta} + 2mn^{2+\delta} + mn + 1\\
                             &\le 6mn^{2+\delta} + mn^{3+2\delta}\\
                             &\le 18n^{3+\delta} + 3n^{4+2\delta}\\
                             &\le 4n^{4+2\delta}.
\end{align*}

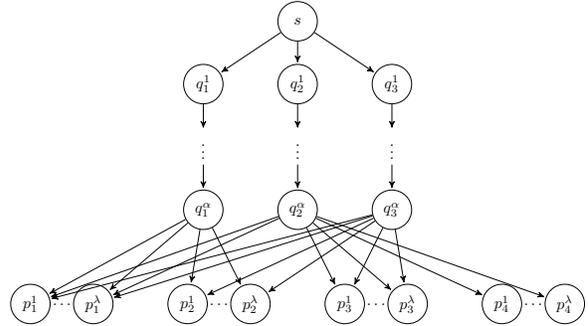
\begin{figure}
    \centering{
        \input{reduction_figure}
    }
    \caption{Illustration of the reduction applied to the Set Cover instance where $X = \{x_i:i\in [4]\}$ and $C = \{\{x_1,x_2\},\{x_1,x_3,x_4\}, \{x_1,x_2,x_3\}\}$.}
    \label{figure:reduction directed}
\end{figure}

We start with the following observation that follows from $\dist(v,s)=\infty$ for any $v\neq s$.
\begin{observation}
    For any $S\subseteq V\setminus\{s\}$, it holds that $\sum_{v\in V\setminus S} \dist(S, v) = \infty$.
\end{observation}

We continue with further observations.
\begin{observation}
    If there is a set cover \(\{C_{j_1},\ldots,C_{j_k}\}\) of size $k$, then $S := \{s\} \cup \{q_{j_r}^{\alpha}|r \in [k]\}$ satisfies
    \begin{equation}\label{eq:yesInst}
        \sum_{v\in V\setminus S} \dist(S, v)
        \le m\alpha^2 + \Lambda n = \Lambda (n+1)
    \end{equation}
\end{observation}
In the above inequality, the first summand is due to the vertices from $\{q_j^l: j \in [m], l\in[\alpha]\}$ that are not in $S$, while the second summand is due to the $\Lambda n$ vertices $\{p_i^t:i\in[n],t\in[\Lambda]\}$ that are all at distance one from $S$.

Conversely, we obtain the following observation.
\begin{observation}
    Let $S \subset \{s\}\cup\{q_{j}^{\alpha}:j\in [m]\}$ be such that $\{C_j|q_{j}^{\alpha}\in S\}$ does not correspond to a set cover in $(X, C, k)$ , then
    \begin{equation}\label{eq:noInst}
        \sum_{v\in V\setminus S} \dist(S, v) \ge \Lambda \alpha.
    \end{equation}
\end{observation}
Note that the previous inequality has made an assumption on the elements that compose the set \(S\). The following lemma, will provide the rational behind this assumption.

\begin{lemma} \label{lemma:directed reduction}
    Let $S$ be a set of vertices containing $s$ such that $S \not\subset \{s\}\cup\{q_{j}^{\alpha}:j\in [m]\}$,
    then either $\{C_j:q_{j}^{\alpha}\in S\}$ corresponds to a set cover in $(X, C, k)$ or we can, in polynomial time, build a set $S'$ from $S$ with $|S'|=|S|$ such that $S\cap \{q_{j}^{\alpha}: j\in [m]\} \subseteq S'\cap \{q_{j}^{\alpha}:j \in [m]\}$ and $\sum_{v\in V\setminus S} \dist(S, v) > \sum_{v\in V\setminus S'} \dist(S', v)$.
\end{lemma}
\begin{proof}
    Let $S$ be a set of vertices containing $s$ such that $S \not\subset \{s\} \cup \{q_{j}^{\alpha}: j\in [m]\}$ and $\{C_j : q_{j}^{\alpha} \in S\}$ is not a set cover in $(X, C, k)$. We distinguish the following (non-exclusive) cases.
    \begin{itemize}
        \item There exists a vertex $p_i^t \in S$ such that $p_i^t$ is not at distance $1$ from $S$. In this case, let $q_{j}^{\alpha}$ be a vertex at distance $1$ of $p_i^t$ and set $S' = S\setminus\{p_i^t\}\cup\{q_{j}^{\alpha}\}$.
        \item There exists a vertex $p_i^t \in S$ such that $p_i^t$ is at distance $1$ from $S$. In this case, because $\{C_j : q_{j}^{\alpha} \in S\}$ is not a set cover of $X$, there exists a vertex $p_{i'}^{t'}$ which is not at distance $1$ from $S$ and a vertex $q_{j}^{\alpha}$ which is at distance $1$ of $p_{i'}^{t'}$. Set $S' = S\setminus\{p_i^t\}\cup\{q_{j}^{\alpha}\}$.
        \item Assume that the two first cases do not occur. Then, there exists a vertex $q_k^\ell$ in $S\setminus \{q_{j}^{\alpha}: j\in [m]\}$. If $q_k^{\alpha} \notin S$ and there exists a vertex $p_i^t$ such that $q_k^{\alpha}$ is closer to $p_i^t$ than any vertex in $S$, then set $S' = S\setminus\{q_k^\ell\}\cup \{q_k^{\alpha}\}$. Otherwise, because $\{C_j: q_{j}^{\alpha} \in S\}$ is not a set cover of $X$, there exists a vertex $p_{i}^{t}$ which is not at distance $1$ from $S$ and a vertex $q_{j'}^{\alpha}$ which is at distance $1$ of $p_{i}^{t}$. Set $S' = S\setminus\{q_k^\ell\}\cup\{q_{j'}^{\alpha}\}$.
    \end{itemize}
    In all cases, it is easy to see that we obtain a set $S'$ which satisfies the conditions of the lemma. In the first two cases, we use the fact that $(\Lambda-k) > 1$. In the third case, the result is due to the fact that $\Lambda$ is greater than the possible loss incurred by vertices $q_j^\ell$.
\hfill\end{proof}

Using Lemma~\ref{lemma:directed reduction}, we can now prove  Theorem~\ref{theorem:hardness directed}.
\begin{proofof}
    Let us assume that there exists an algorithm $A$ with approximation guarantee $4\cdot |V|^{-\epsilon}$ for the group-closeness centrality problem for some $\epsilon \in (0,1/2) $. For a given instance of Set Cover, we use the reduction described above to obtain an instance of the group-closeness centrality problem and apply algorithm $A$ using \(\delta := 4\epsilon/(1-2\epsilon)\). Observe that $\epsilon<1/2$ guarantees that $\delta>0$.
    Let $S$ be the solution returned by $A$. Using Lemma~\ref{lemma:directed reduction}, we can assume that either $\{C_j:q_{j}^{\alpha}\in S\}$ corresponds to a set cover of $X$ or $S\setminus \{q_{j}^{\alpha}|j\in [m]\} = \{s\}$. Let us assume that $\{C_j:q_{j}^{\alpha}\in S\}$ does not correspond to a set cover of $X$. Furthermore, suppose for the purpose of contradiction that the original instance of Set Cover is a YES instance. Then, denoting the optimum to the group-closeness centrality problem problem by $OPT$ and using \eqref{eq:yesInst} and \eqref{eq:noInst} yields that
    \[
        \frac{\GC(S)}{OPT}
        \le \frac{\Lambda (n+1)}{\Lambda \alpha}
        < 2n^{-\delta}
        \le 2\Big(\frac{|V|}{4}\Big)^{-\frac{\delta}{4+2\delta}}
        \le 4|V|^{-\epsilon},
    \]
    using that $|V|\le 4n^{4+2\delta}$ and $\epsilon=\delta/(4+2\delta) < 1/2$.
    This contradicts the assumption that $S$ is a $4\cdot|V|^{-\epsilon}$-approximate solution. 
    To summarize, we have shown that if $A$ is a $4\cdot |V|^{-\epsilon}$-approximation algorithm for the group-closeness centrality problem, then it provides a polynomial time algorithm for Set Cover.\hfill
\end{proofof}

%% file: reduction_figure.tex
\resizebox{.45\textwidth}{!}{ 
    \begin{tikzpicture}[scale=.7,->,>=stealth', shorten >=1pt, auto, semithick, every node/.style = {minimum width = 2.5em}]
        \node[circle,draw] (s)   at (0,6)  {$s$};
        \node[circle,draw] (q1v) at (-3,4) {$q_1^{1}$};
        \node[circle,draw] (q2v) at (0,4) {$q_2^{1}$};
        \node[circle,draw] (q3v) at (3,4) {$q_3^{1}$};
        \node (q1i) at (-3,2) {$\vdots$};
        \node (q2i) at (0,2) {$\vdots$};
        \node (q3i) at (3,2) {$\vdots$};
        \node[circle,draw] (q11) at (-3,0) {$q_1^{\alpha}$};
        \node[circle,draw] (q21) at (0,0) {$q_2^{\alpha}$};
        \node[circle,draw] (q31) at (3,0) {$q_3^{\alpha}$};
        
        \node[circle,draw] (p11) at (-8.5,-3) {$p_1^{1}$};
        \node (p1j) at (-7.5,-3) {$\ldots$};
        \node[circle,draw] (p1lamb) at (-6.5,-3) {$p_1^{\lambda}$};
        
        \node[circle,draw] (p21) at (-3.5,-3) {$p_2^{1}$};
        \node (p2j) at (-2.5,-3) {$\ldots$};
        \node[circle,draw] (p2lamb) at (-1.5,-3) {$p_2^{\lambda}$};
        
        \node[circle,draw] (p31) at (1.5,-3) {$p_3^{1}$};
        \node (p3j) at (2.5,-3) {$\ldots$};
        \node[circle,draw] (p3lamb) at (3.5,-3) {$p_3^{\lambda}$};
        
        \node[circle,draw] (p41) at (6.5,-3) {$p_4^{1}$};
        \node (p4j) at (7.5,-3) {$\ldots$};
        \node[circle,draw] (p4lamb) at (8.5,-3) {$p_4^{\lambda}$};

        \path (s) edge node {} (q1v)
        (s) edge node {} (q2v)
        (s) edge node {} (q3v)
        (q1v) edge node {} (q1i)
        (q2v) edge node {} (q2i)
        (q3v) edge node {} (q3i)
        (q1i) edge node {} (q11)
        (q2i) edge node {} (q21)
        (q3i) edge node {} (q31)
        
        (q11) edge node {} (p11)
        (q11) edge node {} (p1lamb)
        (q21) edge node {} (p11)
        (q21) edge node {} (p1lamb)
        (q31) edge node {} (p11)
        (q31) edge node {} (p1lamb)
        
        (q11) edge node {} (p21)
        (q11) edge node {} (p2lamb)
        (q31) edge node {} (p21)
        (q31) edge node {} (p2lamb)
        
        (q21) edge node {} (p31)
        (q21) edge node {} (p3lamb)
        (q31) edge node {} (p31)
        (q31) edge node {} (p3lamb)
        
        (q21) edge node {} (p41)
        (q21) edge node {} (p4lamb);
      \end{tikzpicture}
    }

%% file: 4_algorithm_engineering.tex
\section{Algorithm Engineering}
\label{sec:algo-eng}

In the following we propose several engineering techniques
that accelerate the approximate maximization of group-closeness
and group-harmonic in practice.


\subsection{Group-Harmonic Maximization.}
\label{sub:ae-group-harm-closeness}
We consider greedy and local search algorithms for group-harmonic centrality.
\paragraph{Greedy Algorithm.}
We start with the greedy algorithm; the pseudocode of this algorithm
is given by Algorithm~\ArxivOrCr{\ref{algo:greedy-group-harmonic} in Appendix~\ref{apx:pseudocodes}}{3 in Appendix H}.
The first vertex that is added to the
group is the vertex with highest harmonic centrality (Line~\ArxivOrCr{\ref{line:top-harmonic-vtx}}{1 of the pseudocode});
this vertex can be found by a top-$k$ algorithm such as
the ones from Refs.~\cite{bisenius2018computing, bergamini2019computing}.
Afterwards, the algorithm iteratively adds the
vertex with highest marginal gain $\GH(S \cup \{u\}) - \GH(S)$ to the group.

Since $\GH$ is submodular, we can evaluate marginal gains lazily, \ie, the
marginal gain $\GHhat(S, u)$ from previous iterations serves as an upper bound
of the marginal gain $\GH(S \cup \{u\}) - \GH(S)$ in the current operation.
Since $\GHhat(S, u) \ge \GH(S \cup \{u\})$ holds after $S$ is initialized with
the vertex with highest harmonic centrality, we initialize $\GHhat(S, u)$ to
H$(u)$ for each $u \in V\setminus S$ (more precisely, the top-$k$ closeness
algorithm from Bisenius \etal~\cite{bisenius2018computing} yields an upper
bound on H$(u)$ that we can use in this initialization step). To determine the
vertex with highest marginal gain, we use the well-known lazy
strategy~\cite{10.1007/BFb0006528}: we evaluate the marginal gain of the vertex
with highest upper bound (and adjust the upper bound to the true marginal gain)
until we know the true marginal gain of the top vertex \wrt the upper bound
(Lines~\ArxivOrCr{\ref{line:group-hclos-pq}}{3}-\ArxivOrCr{\ref{line:max-pq-insert2}}{5} and
Line~\ArxivOrCr{\ref{line:group-hclos-repeat}}{8} of the pseudocode, by using a priority
queue).

To evaluate marginal gains, we run a
pruned SSSP algorithm from $u$ that only visits vertices $v$ such that
$\dist(u, v) < \dist(S, v)$ and updates
$\GHhat(S, u)$ after every vertex at distance $i$
from $u$ has been explored.
The traversal is pruned if $\GHhat(S, u) \le
\GH(S \cup \{x\})$, where $x$ is the vertex with highest marginal gain
computed so far; otherwise it returns the exact value of $\GH(S \cup \{u\})$ once all
that are vertices closer to $u$ than to $S$ have been visited.
As for group-closeness, $\GHhat$ is defined differently for weighted than
for unweighted graphs.

\paragraph{Pruning (Unweighted).} In unweighted graphs, we can exploit additional bounds
to prune the SSSP algorithm earlier.
Let us assume that the pruned SSSP (\iec a BFS) has explored
all vertices up to distance $i$.
We denote by $\Phi_{S, u}^{\le i}$
the set of vertices $v$ such that $\dist(u, v) \le i$
and $\dist(u, v) < \dist(S, v)$. An additional upper bound on the
marginal gain of $u$ is
\begin{align}
\begin{split}
\label{eq:ghc-ub}
&\sum_{v\in \Phi_{S, u}^{\le i} \setminus \{u\}}\left(\frac{1}{\dist(u, v)} -
\frac{1}{\dist(S, v)}\right) +\frac{\tilde{n}_{S, u}^{i + 1}}{i + 1} \\ &+
\frac{\max(0, r(u) - |\Phi_{S, u}^{\le i}| - \tilde{n}_{S, u}^{i + 1})}{i + 2} -
\frac{1}{\dist(S, u)}.
\end{split}
\end{align}

The first term is the contribution of the explored
vertices up to distance $i$ to the marginal gain.
Then, let $\Phi_{S, u}^{i} \subseteq \Phi_{S, u}^{\le i}$
contain the vertices at distance exactly $i$ from $u$; in the second term we assume that
$\tilde{n}_{S, u}^{i + 1} \ge |\Phi_{S, u}^{i + 1}|$ vertices are at distance
exactly $i + 1$ from $u$, where $\tilde{n}_{S, u}^{i + 1}$ is defined
as $\sum_{x\in \Phi_{S, u}^i}\degout(x)$ for directed graphs, and $\sum_{x\in
\Phi_{s, u}^i}(\deg(x) - 1)$ for undirected graphs.
In the third term we assume that all the remaining vertices reachable from
$u$ are at distance $i + 2$ from $u$ (where $r(u)$ is the number of
vertices reachable from $u$\footnote{
Because in directed graphs it is too expensive to compute $r(u)$ for each
vertex, we use an upper bound as described in~\cite{bergamini2019computing}.}).
Finally, we subtract the contribution of $u$ to the centrality of $S$.

As a further optimization for unweighted and undirected graphs, for every vertex $u\in
V\setminus S$ we subtract from $r(u)$ all the vertices in $u$'s connected
component that are at distance 1 from $S$. In this way we avoid to count them
in the third term of Eq.~\eqref{eq:ghc-ub}.

\paragraph{Pruning (Weighted).}
Concerning weighted graphs, the SSSP is a pruned version of the Dijkstra
algorithm. Let $i$ be the distance from $u$ to the last explored vertex.
Upon completion of Dijkstra's relaxation step, $\GHhat(S, u)$ is updated as
follows:

\begin{align}
\label{eq:gh-weighted-bound}
\begin{split}
    \GHhat(S, u) &= \sum_{v\in \Phi_{S, u}^{\le i} \setminus \{u\}}
    \left(\frac{1}{\dist(u, v)} - \frac{1}{\dist(S, v)}\right)\\
                 &+ \frac{r(u) - |\Phi_{S, u}^{\le i}|}{i} - \frac{1}{\dist(S, u)},
\end{split}
\end{align}

\iec to the contribution to $\GHhat(S, u)$ of (i) the vertices visited by
the SSSP, and (ii) the unexplored vertices assuming that they are all at distance
$i$ from $u$.

\paragraph{Local Search.}

The local search algorithm by Lee \etal ~\cite{LMNS10}
needs to evaluate $\Omega(n^2)$ swaps
per iteration. Since this is already quite expensive,
it is desirable to perform only few iterations.
Hence, we initialize the local search with a greedy solution;
this does not affect its approximation guarantee
but offers a considerable acceleration in practice.

We cannot make use of lazy evaluation for local search
(since we need to consider swaps and not vertex additions).
However, we can still make use of the bound from Eq.~\eqref{eq:ghc-ub}.

\paragraph{Parallelism.}

Since both greedy and local search typically need to evaluate
either the marginal gains or the objective function
for many vertices before performing a single
addition (or swap), it is desirable to
utilize parallelism.
We parallelize multiple evaluations of the objective function
in a straightforward way. Each thread evaluates
the marginal gain for one candidate vertex.
It needs to store the state of a single SSSP; this incurs
$\mathcal{O}(n)$ additional memory per thread.


\subsection{Group-Closeness Maximization.}
Since the greedy algorithm for group-closeness has been studied
before~\cite{BergaminiGM18}, we only discuss local-search and engineering improvements.
\paragraph{Local Search.}
We consider a local search algorithm that evaluates all possible pairs of
swaps. For the $k$-Median case, Arya \etal~\cite{AryaGKMMP04} minimize the cost
function of an initial solution $S$; a swap is done only if $cost(S') \le (1 -
\epsilon/Q)\cdot cost(S)$, where $S'$ is the solution after the swap, $Q$ is
the number of neighboring solutions (\ie how many different $S'$ are one swap
away from $S$), and $\epsilon >0$. For group-closeness, the cost function is
represented by $\GF(S)$ (minimum farness is maximum closeness), and $Q =
k\cdot(n - k)$ \iec the number of possible swaps. The algorithm has
an approximation ratio of 5.

Like in the group-harmonic case,
the local search algorithm is much faster in practice if we
start from a good initial solution.
We use the \emph{grow-shrink} algorithm that was introduced
by a subset of the authors~\cite{AngrimanGM19}
to obtain such a solution.
Grow-shrink is a heuristic algorithm;
Ref.~\cite{AngrimanGM19} does not provide any bounds on its approximation
guarantee; however, the paper demonstrates empirically that the algorithm performs
well on real-world graphs.
The lack of approximation ratio in grow-shrink is
not an issue in our case, since the approximation
guarantee of our local search does not depend on the initial solution.


\paragraph{Prioritizing Swaps.}

In practice, the number of swaps that need to be analyzed before a local optimum is reached
is heavily affected by the sequence of swaps that are done.
Algorithm~\ArxivOrCr{\ref{algo:local-search}}{4} in Appendix~\ArxivOrCr{\ref{apx:pseudocodes}}{H} summarizes how we prioritize
the swaps.
Similarly to the original grow-shrink algorithm, we prioritize swaps depending
on their estimated impact on $\GF(S)$.
First, we sort in ascending order the
vertices in $S$ by the increase in $\GF$ due to their removal from $S$ (\iec
$\GF(S \setminus \{u\}) - \GF(S)$ for all $u \in S$, Lines~\ArxivOrCr{\ref{line:pq1}}{4}-\ArxivOrCr{\ref{line:pq2}}{6} of the pseudocode).
Afterwards, we sort in descending order
all the vertices $v \in V\setminus S$ by $\GFapx(\Suv)$, which is an estimate
of the decrease in farness (\iec $\GF(\Suv) - \GF(S)$).
We use the same estimate based on the size of shortest path DAGs as Ref.~\cite{AngrimanGM19}.

As a further optimization, we exclude swaps with vertices in $V\setminus S$
with degree 1 as, in (strongly) connected graphs, they cannot result in a
decrease in $\GF(S)$.


\paragraph{Additional Pruning.}

The grow-shrink algorithm~\cite{AngrimanGM19}
performs pruned SSSPs to evaluate whether a swap is advantageous.
We modify the algorithm to incorporate additional pruning conditions
that prune the SSSP when a swap is not good enough to be
considered in the local search (in contrast, Ref.~\cite{AngrimanGM19}
perform all swaps that improve the objective function, regardless
of the difference in the score).
In particular, we maintain a lower bound $\GFlb(S, u, v) \le \GF(\Suv)$,
so that we can interrupt the pruned SSSP as soon as $\GFlb(S, u, v) > (1 -
\eps/(k\cdot(n - k))) \cdot \GF(S)$.

$\GFlb(S, u, v)$ is computed in two steps: we first compute $\GF^+(u) := \GF(S \setminus \{u\})
- \GF(S)$ exactly (Line~\ArxivOrCr{\ref{line:farn-inc}}{11 of the pseudocode})
\iec the increase in farness
of $S$ due to the removal of $u$.
Then, during every pruned SSSP from $v$, we
keep updating an upper bound of decrease in farness
of $S \setminus \{u\}$ due to the addition of $v$:
$\widehat{\GF}^-(v) := \GF(S \setminus \{u\}) - \GFlb(S, u, v)$. Then, $\GFlb(S, u, v)$ is computed
as $\GF(S) + \GF^+(u) - \widehat{\GF}^-(v)$.

To compute $\GF^+(u)$ exactly we maintain the following information for each
vertex $x\in V\setminus S$: $\dist(S, x)$, a vertex $r_x \in S$ such that
$\dist(r_x, x) = \dist(S,x)$, and $\dist'(S, x) = \dist(S\setminus\{r_x\}, x)$.
In this way, $\GF^+(u)$ can be computed in
$\Oh(n)$ time as done in the original grow-shrink algorithm:

\[
\GF^+(u) = \sum_{x\in \{V\setminus S \st \dist(S, x) = \dist(u,
x)\}}\dist(S, x) - \dist(S', x).
\]

$\widehat{\GF}^-(v)$ is computed differently in unweighted and weighted graphs.
In unweighted graphs the pruned SSSP is a BFS, and we define bounds inspired by the ones
used for top-$k$ closeness centrality in~\cite{bergamini2019computing}:
For every distance $1 \le i \le \diam(G)$ we maintain
$N_{S}^{\ge i}$ \iec the set of vertices at distance $\ge i$ from $S$,
and $\Phi_{S, v}^{\le i}$ \iec the set of vertices $x$ such that $\dist(v, x) \le
i$ and $\dist(v, x) < \dist(S, x)$. Once every vertex in $\Phi_{S,v}^{\le i}$
has been visited by the pruned BFS, we know that at most $\tilde{n}_v^{i +
1} := \min(|N_S^{\ge i + 2}|, \sum_{x\in \Phi_{S, v}^i}\degout(x))$
vertices can be at distance $i + 1$ from $v$ (in undirected graphs,
$\tilde{n}_v^{i + 1} := \min(|N_S^{\ge i + 2}|, \sum_{x\in
\Phi_{S, v}^i}\deg(x) - 1)$)
while the remaining unexplored vertices will be at distance $\ge i + 2$.
Thus, we update $\widehat{\GF}^-(v)$ as follows:
\begin{align*}
\begin{split}
&\widehat{\GF}^-(v) = \sum_{x \in \Phi_{S, v}^{\le i}}(\dist(S, x) - \dist(v, x))\\
                   &+ \sum_{x\in\Lambda}(\dist(S, x) - i - 1)
                   + \sum_{x \in N_S^{\ge i + 3} \setminus \Lambda}(\dist(S,
                   x) - i - 2).
\end{split}
\end{align*}

The first term represents the decrease in farness due to the vertices
that are already
visited by the BFS. In the second term $\Lambda
\subseteq N_S^{\ge i + 2}$ contains the nearest $\tilde{n}_v^{i + 1}$
vertices to $S$, and we assume at they are $i + 1$ hops away from $v$.
Finally, in the third term we assume that all the remaining unvisited vertices
at distance $\ge i + 3$ from $S$ not counted in $\Lambda$ can be reachable
from $v$ in $i + 2$ hops. From the third term we exclude vertices at
distance $i + 2$ from $S$ because, under our assumption, their distance from
$S$ would remain unchanged. At the cost of an additional $\Oh(\diam(G))$
memory, $\widehat{\GF}^-(v)$ can be computed in $\Oh(\diam(G))$ time.

On weighted graphs we update $\widehat{\GF}^-(v)$
by adapting the our strategy from $\GH$ to $\GF$
(see Eq.~\eqref{eq:gh-weighted-bound}).

\paragraph{Parallelism.}
We employ the same parallelism as for group-harmonic centrality.
The fact that evaluations of the objective function can be parallelized
in the greedy and local search algorithm can be seen as an advantage
over the grow-shrink algorithm since the latter
operates inherently sequentially (\ie in many cases, it performs
swaps after evaluating the objective function only once,
even if this does not lead to an improvement in the objective function).

%% file: 5_experiments.tex
\section{Experiments}
\label{sec:exp}

We conduct experiments to evaluate our algorithms in terms of solution quality
and running time.

For $\GH$ we first evaluate the quality of our greedy algorithm (\greedyh),
our local-search algorithm (\grlsh), and sets of vertices
selected uniformly at random (\rndh, the best of 100 randomly chosen sets)
against the optimal solution on small-sized networks.
Then, we measure the quality and running time performance of \greedyh and \grlsh and we use
\rndh as baseline.

Regarding group closeness, we compare our local-search algorithm against the greedy
algorithm by Bergamini \etal~\cite{BergaminiGM18}, the grow-shrink algorithm\footnote{
As in Ref.~\cite[Sec. III.B]{AngrimanGM19}, we use a variant of this algorithm that achieves
a reasonable time-quality trade-off \ie with $p = 0.75$.},
and vertices selected uniformly at random (again, the best of 100 randomly chosen sets).
Hereafter, these algorithms are referred to as \greedy, \gs, and \rnd respectively.
Our local-search algorithm for group closeness uses either \greedy or \gs to
initialize the initial group: in the former case we label it as \grls, and
\gsls in the latter.

\begin{table}[tb]
\footnotesize
\centering
\input{tables/harmonic-closeness-small-diameter-large}
\caption{Large complex networks. The \quot{Type} column
indicates whether the network is undirected (\texttt{U}) or
directed (\texttt{D}).}
\label{tab:cplx-harmonic-large}
\end{table}

\begin{table}[tb]
\footnotesize
\centering
\input{tables/harmonic-closeness-high-diameter-large}
\caption{Large high-diameter networks. In the \quot{Type} column
the first letter indicates whether the network is undirected (\texttt{U})
or directed (\texttt{D}), while the second letter whether the network is unweighted (\texttt{U})
or weighted (\texttt{W}).}
\label{tab:high-diam-harmonic-large}
\end{table}

\subsection{Settings.}
We implemented all algorithms in C++, using the
NetworKit~\cite{staudt2016networkit} graph APIs, and we use
SCIP~\cite{GamrathEtal2020ZR} to solve ILP instances. All experiments are
conducted on a Linux machine with an Intel Xeon Gold 6126 CPU (2 sockets, 12 cores
each) and 192 GB of RAM, and managed by the
SimexPal~\cite{angriman2019guidelines} software to ensure reproducibility.
When averaging approximation ratios and speedups, we use the geometric mean.
All experiments have a timeout of one hour.

Experiments are executed on real-world complex and high-diameter networks.
Sources and detailed statistics of our datasets are reported in
Appendix~\ArxivOrCr{\ref{apx:insts-stats}}{G}.

\subsection{Instances Statistics}
\paragraph{Data Sets.}
Instances have been downloaded from the public repositories
KONECT~\cite{kunegis2013konect}, OpenStreetMap~\cite{OpenStreetMap} (from which
we build the car routing graph using
RoutingKit~\cite{dibbeltSW16}), and from
the 9th DIMACS Implementation Challenge~\cite{demetrescu2009shortest}.
Small instances used for the experiments with ILP solvers are reported in
\ArxivOrCr{Tables~\ref{tab:cplx-small} and~\ref{tab:high-diam-small} in Appendix~\ref{apx:insts-stats}}{Tables 6 and~7}, while the rest of the experiments have been conducted on
the instances in Tables~\ref{tab:cplx-harmonic-large}
and~\ref{tab:high-diam-harmonic-large}.

Because algorithms for group-closeness maximization can only handle (strongly)
connected graphs, we run them on the (strongly) connected components of the
instances in our datasets. Detailed statistics are reported in
Appendix~\ArxivOrCr{\ref{apx:insts-stats}}{G}.
For high-diameter networks we test mainly road networks because
they are the most common type of networks in the aforementioned repositories.
We are confident that our local-search algorithms are capable to handle
other types of high-diameter networks as well without significant difference
in performance.
Because public repositories do not provide a reasonable amount of weighted
complex networks, we omit these networks from our experiments.

\input{plots-ghc-exact}
\subsection{Group Harmonic Maximization.}
\label{sub:exp-ghc}
\paragraph{Comparison to Exact ILP Solutions.}
Figure~\ref{fig:quality-harmonic-ilp}
shows a comparison of the solution quality of our algorithms
compared to exact solutions.
We observe that random groups
cover these unweighted graphs reasonably well;
hence, \rndh already yields solutions of $>70\%$ of the optimum.
This peculiarity is amplified by the fact that the networks
are rather small in comparison to $k$ (at most 1000 vertices).
Indeed, the quality of \rndh \emph{increases} with $k$ on
complex networks,
a behavior that no other algorithm shows.
Still, \greedyh yields substantially better solutions in all cases:
it yields solutions of $>99.5\%$ of the optimum for
all group sizes. These solutions are further improved by \grlsh, which yields
groups with at least $\minQualLSGRHCplxUnw$ of the optimal quality.

In high-diameter networks (Figures~\ref{fig:quality-harmonic-ilp-high-diameter}
and~\ArxivOrCr{\ref{fig:apx:quality-harmonic-ilp} in Appendix~\ref{apx:exact-harmonic})}{7 in Appendix~E} \rndh
is not a serious competitor. It yields solutions
less than $80\%$ the optimal quality. Indeed, since high-diameter
networks have a higher diameter compared to complex networks,
it is expected that a random group of vertices is less likely to be central.
On the other hand, \greedyh and \grlsh yield
solution qualities from $\minQualGRHRoadUnw$
and $\minQualLSGRHRoadUnw$, respectively.
For $k = 5$ in particular,
solutions returned by \grlsh have $>99.99\%$ the quality of the optimal
solution.

Concerning weighted high-diameter networks, the ILP solver runs out of time
or memory on nearly all instances. Tentative results on the two remaining
instances suggest that \greedyh yields solutions at are almost optimal,
but due to the small size of the data set, we cannot conclude definitive results.

\paragraph{Quality and Running Time on Larger Instances.}

\input{plots-ghc-large}
Figure~\ref{fig:quality-harmonic} summarizes quality and running time results of
\greedyh and \grlsh (absolute running times are reported in Tables~\ArxivOrCr{\ref{tab:runtime-h-cplx}
and~\ref{tab:runtime-h-high-diam}, Appendix~\ref{sec:running-times}}{10 and~11 in Appendix I}).
Due to the size of these graphs, it is not feasible to obtain an ILP solution
and we use \rndh as baseline.
In unweighted complex networks (Figures~\ref{fig:quality-harmonic-cplx}
and~\ArxivOrCr{\ref{fig:apx:quality-harmonic-cplx} in Appendix~\ref{apx:larger-harmonic}}{8 in Appendix~E}),
\greedyh finds solutions with quality (compared to \rndh)
from \minQualGRHRndCplxDir (with $k = 5$) to \maxQualGRHRndCplxDir (with
$k = 50$) in directed networks, and from \minQualGRHRndCplxUnd to
\maxQualGRHRndCplxUnd in undirected networks.
Compared to \greedyh, \grlsh is not competitive: it improves the quality by at most
\maxQualImprLSGRHCplxDir while being \minSlowdLSGRHCplxDir to \maxSlowdLSGRHCplxDir
slower.

\greedyh achieves even better results in high-diameter networks: in weighted
directed high-diameter networks
(Figure~\ref{fig:quality-harmonic-high-diameter}) \greedyh's quality is
\minQualGRHRndRoadDirWei to \maxQualGRHRndRoadDirWei of the quality returned by
\rndh, while being just \minSlowdGRHRoadDirWei to \maxSlowdGRHRoadDirWei
slower.
Concerning \grlsh, it is less competitive than in complex networks: it improves
\greedyh's quality by at most \maxQualImprLSGRHRoadDirWei, while being
\minSlowdLSGRHRoadDirWei to \maxSlowdLSGRHRoadDirWei slower. Results are more
promising in high-diameter unweighted networks
(Figure~\ArxivOrCr{\ref{fig:apx:quality-harmonic-high-diameter} in
Appendix~\ref{apx:larger-harmonic}}{9 in Appendix E}): here \grlsh improves \greedyh's quality by
\minQualImprLSGRHRoadUnw to \maxQualImprLSGRHRoadUnw while being \minSlowdGRHRoadUnw
to \maxSlowdGRHRoadUnw slower.

\subsection{Group Closeness Maximization.}
\paragraph{Comparison to Exact ILP Solutions.}
\input{plots-gc-exact}

Figures~\ref{fig:quality-closeness-ilp} and~\ArxivOrCr{\ref{fig:apx:quality-closeness-ilp}
(Appendix~\ref{apx:add-exp-clos})}{10 (Appendix F)} summarize the quality of our local-search
algorithms for group closeness and the competitors compared to the optimum.

Concerning unweighted complex networks,
in the directed case, for groups of size 5 \grls is the only algorithm achieving optimal solutions,
while for the remaining group sizes it yields solutions with the same quality as \greedy.
In the undirected case (see Figure~\ArxivOrCr{\ref{fig:apx:quality-closeness-ilp}}{10})
\grls and \gsls achieve solutions with at least \minQualLSGRCplxUnw
and \minQualLSGSCplxUnw the optimal quality, resp.; for $k = 5$ and $k =
100$ in particular they achieve optimal solutions.

In high-diameter networks our local-search algorithms always achieve better results than
\greedy and \gs. The best results are on unweighted graphs: here \grls and \gsls yield
solutions at least \minQualLSGRRoadUnw and \minQualLSGSRoadUnw away from
optimality, respectively.

Interestingly, the quality of \grls is often higher
than \gsls, especially in complex networks and high-diameter weighted networks; we
conjecture that our local-search algorithm has a narrower improvement margin on
solutions from \gs compared to solutions from \greedy
since \gs is based on local-search as well.

\subsection{Quality and Running Time on Larger Instances.}
\input{plots-gc-large}
In Figures~\ref{fig:quality-closeness} and~\ArxivOrCr{\ref{fig:apx:quality-closeness}
(Appendix~\ref{apx:add-exp-clos})}{11 (Appendix F)}, we report the quality and running time
results of \gsls, \grls, \greedy and \gs compared to \rnd (absolute running times are reported in
Tables~\ArxivOrCr{\ref{tab:runtime-c-cplx}
and~\ref{tab:runtime-c-high-diam}, Appendix~\ref{sec:running-times}}{12 and 13, Appendix I}).
In terms of quality our local-search algorithms always reach the best results in
all our experiments:
in directed complex networks (Figure~\ref{fig:quality-cplx}) \gsls, \grls, and \greedy yield similar quality,
while \gs has consistently the lowest quality. On the other hand, quality
can be traded for running time: \gs is the fastest algorithm (for small group
sizes even faster than \rnd), \greedy is on average \avgSpeedGreedyCplxUnw
slower than \rnd (average among all $k$), whereas \gsls and \grls are
respectively \minSpeedGSLSCplxUnw to \maxSpeedGSLSCplxUnw, and
\minSpeedGRLSCplxUnw to \maxSpeedGRLSCplxUnw slower than \rnd. Interestingly,
for small group sizes \grls is often faster
than \gsls,
and vice versa for larger groups.
This is likely due to the difference between \gs and \greedy solutions:
\greedy aims to maximize the objective function regardless of the group size,
while for \gs the group size determines how many vertices are consecutively added
and removed in a single iteration. Therefore, for larger groups \gs solutions need
less swaps to reach a local optimum compared to \greedy solutions.

In directed weighted high-diameter networks (Figure~\ref{fig:quality-high-diam})
\grls always achieves the highest quality with less time overhead
than \gsls for all group sizes but $100$.


\subsection{Parallel Scalability.}
\label{apx:parallel-scalability}

Strong scaling plots for \greedy, \grls, and \gs are reported in
Figure~\ref{fig:par-scal}.
On average our local search algorithms scale better than \greedy
on both complex and high-diameter networks.
This is not surprising: local search needs to evaluate at least $k(n - k)$
swaps which is a highly parallel operation, and often much more expensive
than running \greedy.

On high-diameter networks in particular, \greedy has a poor parallel scalability;
we conjecture that, since closeness centrality distinguishes vertices in
high-diameter networks better than in complex
networks~\cite[Ch.~7]{newman2018networks}, \greedy needs to evaluate only few
vertices per iteration before finding the vertex with highest marginal gain.,
In that case, multiple cores do not speed this process up significantly.

\input{plots-parallel-scal}

%% file: tables/harmonic-closeness-small-diameter-large.tex
\begin{tabular}{lcrr}
\midrule
Graph & Type & $|V|$ & $|E|$\\
\midrule
petster-hamster-household & \texttt{U} & \numprint{874} & \numprint{4003}\\
petster-hamster-friend & \texttt{U} & \numprint{1788} & \numprint{12476}\\
petster-hamster & \texttt{U} & \numprint{2000} & \numprint{16098}\\
loc-brightkite\_edges & \texttt{U} & \numprint{58228} & \numprint{214078}\\
douban & \texttt{U} & \numprint{154908} & \numprint{327162}\\
petster-cat-household & \texttt{U} & \numprint{105138} & \numprint{494858}\\
loc-gowalla\_edges & \texttt{U} & \numprint{196591} & \numprint{950327}\\
wikipedia\_link\_fy & \texttt{U} & \numprint{65562} & \numprint{1071668}\\
wikipedia\_link\_ckb & \texttt{U} & \numprint{60722} & \numprint{1176289}\\
petster-dog-household & \texttt{U} & \numprint{260390} & \numprint{2148179}\\
livemocha & \texttt{U} & \numprint{104103} & \numprint{2193083}\\
flickrEdges & \texttt{U} & \numprint{105938} & \numprint{2316948}\\
petster-friendships-cat & \texttt{U} & \numprint{149700} & \numprint{5448197}\\
\midrule
wikipedia\_link\_mi & \texttt{D} & \numprint{7996} & \numprint{116457}\\
foldoc & \texttt{D} & \numprint{13356} & \numprint{120238}\\
wikipedia\_link\_so & \texttt{D} & \numprint{7439} & \numprint{125046}\\
wikipedia\_link\_lo & \texttt{D} & \numprint{3811} & \numprint{132837}\\
wikipedia\_link\_co & \texttt{D} & \numprint{8252} & \numprint{177420}\\
\midrule
\end{tabular}

%% file: tables/harmonic-closeness-high-diameter-large.tex
\begin{tabular}{lcrr}
\midrule
Graph & Type & $|V|$ & $|E|$\\
\midrule
marshall-islands & \texttt{UU} & \numprint{1080} & \numprint{2557}\\
micronesia & \texttt{UU} & \numprint{1703} & \numprint{3600}\\
kiribati & \texttt{UU} & \numprint{1867} & \numprint{4412}\\
opsahl-powergrid & \texttt{UU} & \numprint{4941} & \numprint{6594}\\
samoa & \texttt{UU} & \numprint{6926} & \numprint{15217}\\
comores & \texttt{UU} & \numprint{7250} & \numprint{17554}\\
\midrule
marshall-islands & \texttt{UW} & \numprint{1080} & \numprint{2557}\\
micronesia & \texttt{UW} & \numprint{1703} & \numprint{3600}\\
kiribati & \texttt{UW} & \numprint{1867} & \numprint{4412}\\
DC & \texttt{UW} & \numprint{9522} & \numprint{14807}\\
samoa & \texttt{UW} & \numprint{6926} & \numprint{15217}\\
comores & \texttt{UW} & \numprint{7250} & \numprint{17554}\\
\midrule
marshall-islands & \texttt{DU} & \numprint{1080} & \numprint{2557}\\
micronesia & \texttt{DU} & \numprint{1703} & \numprint{3600}\\
kiribati & \texttt{DU} & \numprint{1867} & \numprint{4412}\\
samoa & \texttt{DU} & \numprint{6926} & \numprint{15217}\\
comores & \texttt{DU} & \numprint{7250} & \numprint{17554}\\
opsahl-openflights & \texttt{DU} & \numprint{2939} & \numprint{30501}\\
tntp-ChicagoRegional & \texttt{DU} & \numprint{12982} & \numprint{39018}\\
\midrule
marshall-islands & \texttt{DW} & \numprint{1080} & \numprint{2557}\\
micronesia & \texttt{DW} & \numprint{1703} & \numprint{3600}\\
kiribati & \texttt{DW} & \numprint{1867} & \numprint{4412}\\
samoa & \texttt{DW} & \numprint{6926} & \numprint{15217}\\
comores & \texttt{DW} & \numprint{7250} & \numprint{17554}\\
\midrule
\end{tabular}

%% file: plots-ghc-exact.tex
\begin{figure}[tb]
\centering
\begin{subfigure}[t]{\columnwidth}
\centering
\includegraphics{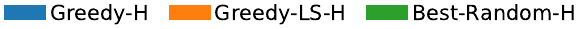}
\end{subfigure}\smallskip

\begin{subfigure}[t]{.5\columnwidth}
\centering
\includegraphics{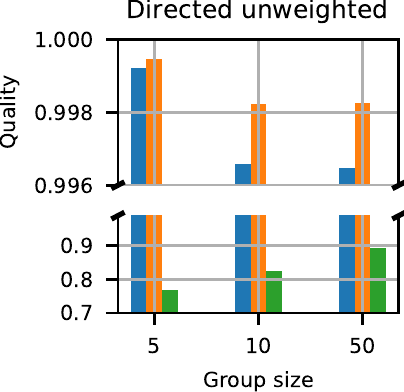}
\caption{Complex networks}
\label{fig:quality-harmonic-ilp-cplx}
\end{subfigure}\hfill
\begin{subfigure}[t]{.5\columnwidth}
\centering
\includegraphics{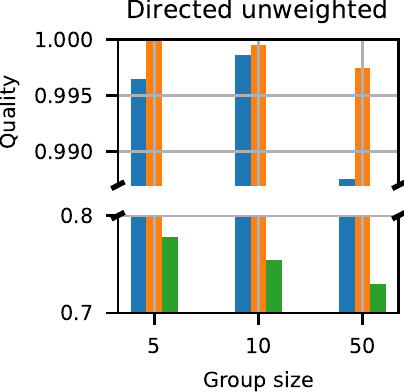}
\caption{High-diameter networks}
\label{fig:quality-harmonic-ilp-high-diameter}
\end{subfigure}

\caption{Quality vs. the optimum over the networks of Tables~\ArxivOrCr{\ref{tab:cplx-harmonic-small}
and~\ref{tab:high-diam-harmonic-small}, Appendix~\ref{apx:insts-stats}}{4 and 5, Appendix G}.}
\label{fig:quality-harmonic-ilp}
\end{figure}

%% file: plots-ghc-large.tex
\begin{figure}[tb]
\centering
\begin{subfigure}[t]{\columnwidth}
\centering
\includegraphics{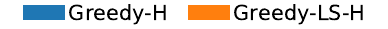}
\end{subfigure}\smallskip

\begin{subfigure}[t]{\columnwidth}
\centering
\begin{subfigure}[t]{.5\columnwidth}
\centering
\includegraphics{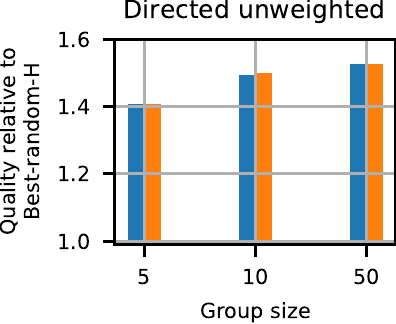}
\end{subfigure}\hfill
\begin{subfigure}[t]{.5\columnwidth}
\centering
\includegraphics{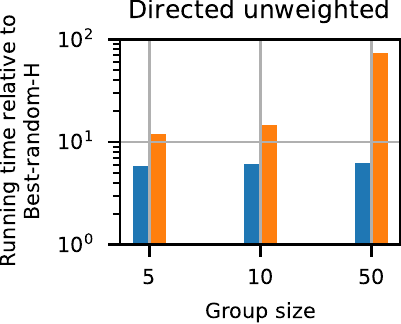}
\end{subfigure}
\caption{Complex networks}
\label{fig:quality-harmonic-cplx}
\end{subfigure}\smallskip

\begin{subfigure}[t]{\columnwidth}
\centering
\begin{subfigure}[t]{.5\columnwidth}
\centering
\includegraphics{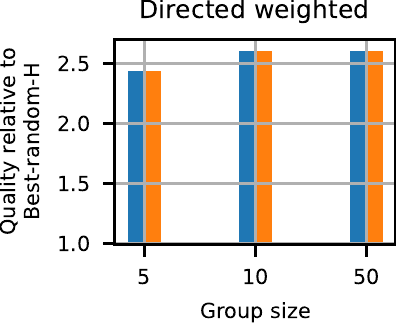}
\end{subfigure}\hfill
\begin{subfigure}[t]{.5\columnwidth}
\centering
\includegraphics{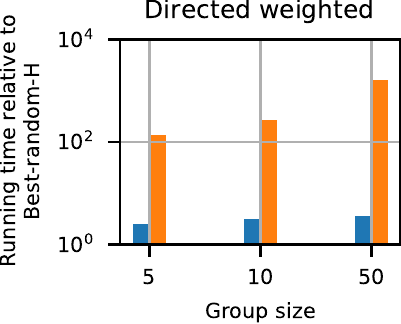}
\end{subfigure}
\caption{High-diameter networks}
\label{fig:quality-harmonic-high-diameter}
\end{subfigure}

\caption{Quality and time \wrt \rndh over the networks of Tables~\ref{tab:cplx-harmonic-large}
and~\ref{tab:high-diam-harmonic-large}.}
\label{fig:quality-harmonic}
\end{figure}

%% file: plots-gc-exact.tex
\begin{figure}[tb]
\centering
\begin{subfigure}[t]{\columnwidth}
\centering
\includegraphics{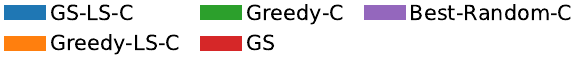}
\end{subfigure}\smallskip

\begin{subfigure}[t]{.5\columnwidth}
\centering
\includegraphics{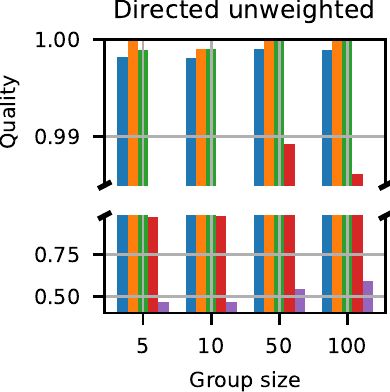}
\caption{Complex networks}
\label{fig:quality-ilp-cplx}
\end{subfigure}\hfill
\begin{subfigure}[t]{.5\columnwidth}
\centering
\includegraphics{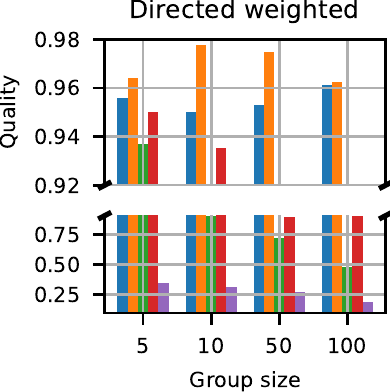}
\caption{High-diameter networks}
\label{fig:quality-ilp-high-diam}
\end{subfigure}

\caption{Quality \wrt the optimum over the networks of Tables~\ArxivOrCr{\ref{tab:cplx-small}
and~\ref{tab:high-diam-small}, Appendix~\ref{apx:insts-stats}}{6 and 7, Appendix G}.}
\label{fig:quality-closeness-ilp}
\end{figure}

%% file: plots-gc-large.tex
\begin{figure}[tb]
\centering
\begin{subfigure}[t]{\columnwidth}
\centering
\includegraphics{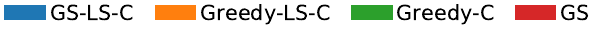}
\end{subfigure}\smallskip

\begin{subfigure}[t]{\columnwidth}
\begin{subfigure}[t]{.5\columnwidth}
\centering
\includegraphics{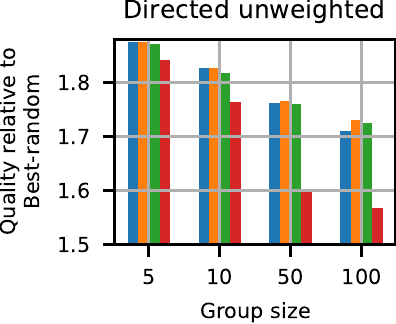}
\end{subfigure}\hfill
\begin{subfigure}[t]{.5\columnwidth}
\centering
\includegraphics{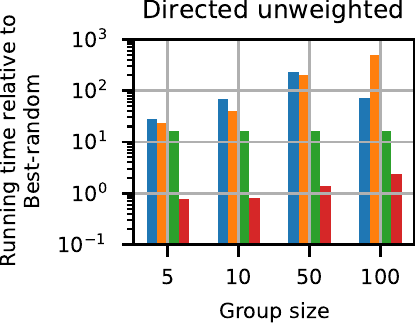}
\end{subfigure}
\caption{Complex networks}
\label{fig:quality-cplx}
\end{subfigure}\smallskip

\begin{subfigure}[t]{\columnwidth}
\begin{subfigure}[t]{.5\columnwidth}
\centering
\includegraphics{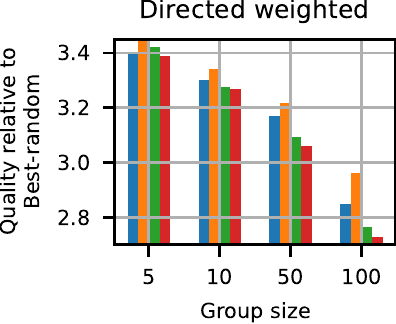}
\end{subfigure}\hfill
\begin{subfigure}[t]{.5\columnwidth}
\centering
\includegraphics{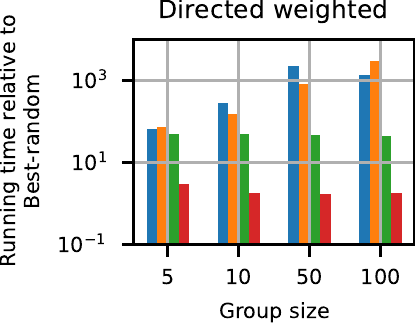}
\end{subfigure}
\caption{High-diameter networks}
\label{fig:quality-high-diam}
\end{subfigure}

\caption{Quality and time \wrt \rnd over the networks of Tables~\ArxivOrCr{\ref{tab:cplx-large}
and~\ref{tab:high-diam-large}, Appendix~\ref{apx:insts-stats}}{8 and 9, Appendix G}.}
\label{fig:quality-closeness}
\end{figure}

%% file: plots-parallel-scal.tex
\begin{figure}[tb]
\centering
\includegraphics{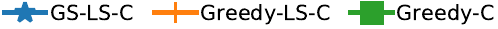}\smallskip

\begin{subfigure}[t]{\columnwidth}
\begin{subfigure}[t]{.5\columnwidth}
\centering
\includegraphics{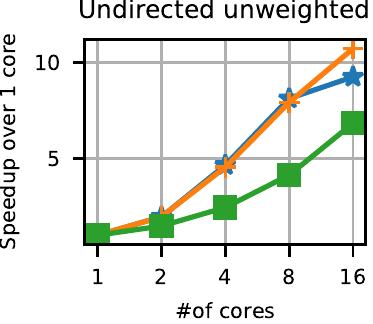}
\end{subfigure}%
\begin{subfigure}[t]{.5\columnwidth}
\centering
\includegraphics{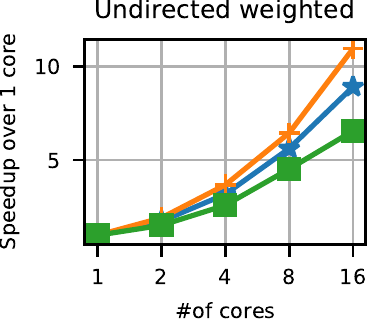}
\end{subfigure}
\caption{Complex networks}
\label{fig:par-scal-cplx}
\end{subfigure}\medskip

\begin{subfigure}[t]{\columnwidth}
\begin{subfigure}[t]{.5\columnwidth}
\centering
\includegraphics{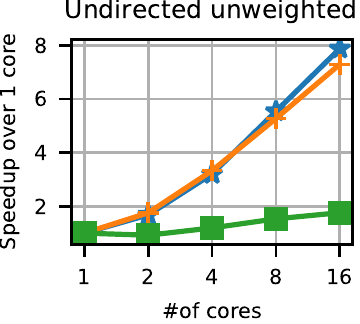}
\end{subfigure}\hfill
\begin{subfigure}[t]{.5\columnwidth}
\centering
\includegraphics{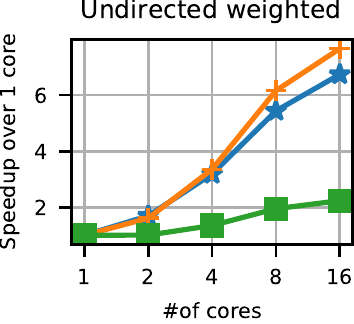}
\end{subfigure}
\caption{High-diameter networks}
\end{subfigure}\medskip

\begin{subfigure}[t]{\columnwidth}
\begin{subfigure}[t]{.5\columnwidth}
\centering
\includegraphics{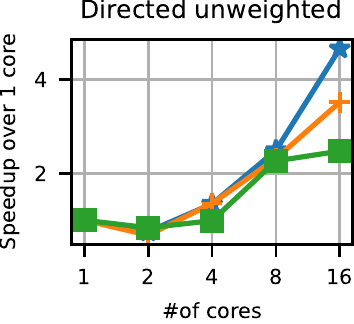}
\end{subfigure}\hfill
\begin{subfigure}[t]{.5\columnwidth}
\centering
\includegraphics{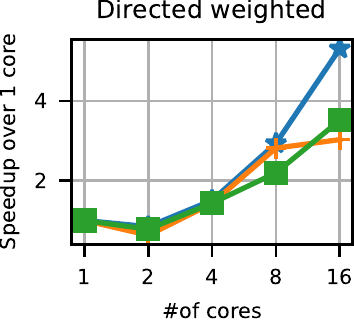}
\end{subfigure}
\caption{High-diameter networks}
\end{subfigure}

\caption{Parallel scalability over the networks of Tables~\ArxivOrCr{\ref{tab:cplx-large}
and~\ref{tab:high-diam-large}, Appendix~\ref{apx:insts-stats}}{8 and 9, Appendix G}.}
\label{fig:par-scal}
\end{figure}

%% file: 6_conclusion.tex
\section{Conclusions}

This work has investigated theoretical and practical approximation aspects of two group centrality maximization problems, namely group-harmonic maximization and group-closeness maximization. 
These two problems aim to determine a group of $k$ nodes in a network which is central as a whole.

For the first problem, we have provided approximation hardness results as well as interesting approximation guarantees for a local search algorithm and the greedy algorithm.
For the second one, we showed that the undirected version of the problem
admits a constant-factor approximation algorithm, while the directed version is $NP$-hard to approximate better than $\Theta(|V|^{-\epsilon})$ for any $\epsilon < 1/2$. 
We have illustrated how to implement efficiently greedy and local search heuristics for both
problems and presented the results of a detailed experimental study. 
Our experiments show that the quality of both the greedy and the local search algorithms come very close to the optimum. 
This finding is consistent with our theoretical results which assess that in most cases these algorithms have good approximation guarantees. 
Interestingly, the two methods also perform well on directed instances for group-closeness maximization despite the hardness of approximation result which holds on this class of instances.  

Several future works are conceivable. First, one could try to close the gap on group-harmonic maximization between existing approximation guarantees and approximation hardness results. 
Second, for group-closeness maximization, it would be interesting to design an algorithm with an approximation ratio matching our hardness result in the directed case.

%% file: appendix.tex
\appendix

\section{Omitted proofs}
\subsection{Proof of Lemma~\ref{harmonic:submodularity}.}

We show that for any $S\subseteq T \subseteq V$ and $v\in V\setminus T$ the following holds:
\[
\GH(S\cup \{v\}) - \GH(S) \geq \GH(T\cup \{v\}) - \GH(T).
\]
The LHS is equal to
\begin{align}\label{harmonic:submodularity:lhs}
    \begin{split}
        &\underbrace{\sum_{u\in T \setminus S}\left( \frac{1}{\dist(S\cup\{v\},u)} - \frac{1}{\dist(S,u)} \right)}_{(a)}
        \underbrace{- \frac{1}{\dist(S,v)}}_{(b)}\\
        &\quad+ \underbrace{ \sum_{u\in V \setminus (T\cup \{v\})}\left( \frac{1}{\dist(S\cup\{v\},u)} - \frac{1}{\dist(S,u)} \right)}_{(c)},
    \end{split}
\end{align}
while the RHS is equal to
\begin{equation}\label{harmonic:submodularity:rhs}
\underbrace{- \frac{1}{\dist(T,v)}}_{(b')}
+ \underbrace{\sum_{u\in V \setminus (T\cup \{v\})}\left( \frac{1}{\dist(T\cup\{v\},u)} - \frac{1}{\dist(T,u)} \right)}_{(c')}.
\end{equation}
Term $(a)$ in~\eqref{harmonic:submodularity:lhs} is non-negative, term $(b)$ is at least equal to term $(b')$, we show that term $(c)$ is at least $(c')$. We analyze each term of the sum, $u\in V \setminus (T\cup \{v\})$, separately, we have two cases: 
(1) $\dist(T\cup\{v\},u) = \dist(T,u)$. In this case the term related to $u$ in $(c')$ is equal to 0, while the one in $(c)$ is non-negative.
(2) $\dist(T\cup\{v\},u) < \dist(T,u)$. In this case we have $\dist(v,u) < \dist(T,u)\leq \dist(S,u)$ and  $\dist(S\cup\{v\},u)=\dist(T\cup\{v\},u)$. It follows that 
    $
     \frac{1}{\dist(S\cup\{v\},u)} - \frac{1}{\dist(S,u)} 
     = \frac{1}{\dist(T\cup\{v\},u)} - \frac{1}{\dist(S,u)}
     \geq \frac{1}{\dist(T\cup\{v\},u)} - \frac{1}{\dist(T,u)}
    $, which concludes the proof.

\subsection{Proof of Lemma~\ref{harmonic:lemmaweighted}.}

Let $A$ be an $\alpha$-approximation algorithm for the unweighted case of the problem. 
Given an instance $I$ of the group-harmonic centrality problem, we denote by $I_{\overline{w}}$ its unweighted version (setting all weights to 1). 
We denote by $\GH(\cdot)$ and $\GH_{\overline{w}}(\cdot)$ the corresponding group-harmonic objective functions and let $O_{\overline{w}}$ and $O$ be optimal solutions in $I_{\overline{w}}$ and $I$, respectively. 
Apply algorithm $A$ to $I_{\overline{w}}$ and let $S$ be the returned solution. 
We have that $\GH_{\overline{w}}(S) \ge \alpha \GH_{\overline{w}}(O_{\overline{w}}) \ge \alpha \GH_{\overline{w}}(O)$. 
Moreover, for any set $T$, it is easy to observe that $\GH_{\overline{w}}(T) \times \frac{1}{\LMAX} \le \GH(T) \le \GH_{\overline{w}}(T) \times \frac{1}{\LMIN}$. Hence, we obtain that $\GH(S) \ge \alpha\lambda\GH(O)$.

\subsection{Proof of Lemma~\ref{lemma: one half approx}.}
In the directed (resp. undirected) case, let $i$ be the first iteration such that $\Delta_i < 0$ (resp. $\Delta_i \le 0$). 
We show that in this case, for each $v\in V\setminus S_{i-1}$, $\dist(S_{i-1},v) \le 2$.
By contradiction, let us consider a node $v$ such that $\dist(S_{i-1},v) \ge 3$.
We first argue that $\dist(S_{i-1},v)  < \infty$.  
Indeed, in the directed case, if $\dist(S_{i-1},v)  = \infty$, then $v$ would yield a non-negative increment. 
Moreover, in the undirected case, if $\dist(S_{i-1},v)  = \infty$, then $v$ would yield a positive increment as we have assumed there are no isolated nodes in $G$. 
Hence, $\dist(S_{i-1},v)  < \infty$ and there exists a neighbor $u$ of $v$ on a shortest from $S_{i-1}$ to $v$. Then, $\dist(S_{i-1},u) \ge 2$,
$\dist(S_{i-1}\cup\{u\},v) = 1$
and $\GH(S_{i-1} \cup \{u\}) - \GH(S_{i-1}) \geq - \frac{1}{\dist(S_{i-1},u)} 
+ \frac{1}{\dist(S_{i-1}\cup\{u\},v)} - \frac{1}{\dist(S_{i-1},v)}
\ge - \frac{1}{2} + 1 - \frac{1}{3} > 0$, 
a contradiction to $\Delta_i < 0$ (resp. $\Delta_i \le 0$). 

Let $S$ be the set returned by Algorithm~\ref{algo:mon+submodular:greedy}. 
As $S_{i-1}\subset S$, it follows that the group-harmonic centrality of $S$ can be lower-bounded by $\frac{|V|-k}{2}$,
while the optimum can be upper-bounded by $|V| - k$. 
Thus, the approximation ratio guaranteed by $S$ is at least $0.5$.

\subsection{Proof of Theorem~\ref{harmonic:hardness:undirected}.}\label{apx:reduction harmonic undirected}

\begin{algorithm}
\caption{Approximation algorithm for Minimum Dominating Set used in the proof of Theorem~\ref{harmonic:hardness:undirected}}
\label{algo:harmonic:hardness}
\begin{algorithmic}[1]
\State We assume that there exists a $\gamma$-approximation algorithm $\mathcal{A}$ for the group harmonic maximization problem.
\For {$k=1,\ldots,n$}
  \State $D_k \gets \emptyset$
  \State $V_k^1\gets V$
  \State $j \gets 1$
  \While {$|V_k^j|\geq k +1$}
    \State Let $n_k^j=|V_k^j|$ and assume w.l.o.g. that $V_k^j = [n_k^j]$
    \State Build a graph $G_k^j=(\bar{V}_k^j,E_k^j)$ from the subgraph of $G$ induced by $V_k^j$ as follows
      \State $\bar{V}_k^j\gets V_k^j \cup\{x\}\cup Y_k^j \cup Z_k^j$, where  $Y_k^j:=\{y_i~|~i=1,\ldots,k\}$ and $Z_k^j:=\{z_i~|~i=1,\ldots, n_k^j\}$
      \State $E_k^j\gets E(V_k^j) \cup \{\{x,y_i\}~|~y_i\in Y_k^j\} \cup \{\{x,z_i\}~|~ z_i\in Z_k^j\} \cup \{\{z_i,{i}\}~|~i=1,\ldots, n_k^j\}  $
    \State Let $S_k^j$ be the solution returned by algorithm $\mathcal{A}$ on $G_k^j$ with budget $k+1$
    \State $D_k\gets D_k \cup (S_k^j \cap V_k^j)$
    \State $V_k^{j+1} \gets V_k^{j} \setminus (S_k^j \cup \bigcup_{v\in S_k^j}N_v)$
    \State $j\gets j+1$
  \EndWhile
  \State $D_k\gets D_k \cup V_k^j$
\EndFor
\State $D\gets \argmin_{k=1,\ldots,n} |D_k|$
\State\Return $D$
\end{algorithmic}
\end{algorithm}

By contradiction, let us assume that there exists a $\gamma$-approximation algorithm $\mathcal{A}$ for the group harmonic maximization problem, where $\gamma > 1-1/4e$. We show that, using $\mathcal{A}$, Algorithm~\ref{algo:harmonic:hardness} is a logarithmic-factor approximation algorithm to the minimum dominating set problem, which is a contradiction, unless $P=NP$~\cite{DinurS14}.

The \emph{Minimum Dominating Set} problem is defined as follows, let $G=(V,E)$ be an undirected graph, where $V=[n]$, find a dominating set, i.e. a set of nodes $D\subset V$ such that $V=D\cup \bigcup_{v\in D}N_v$, of minimum size. For any $c\in(0,1)$, there exist no $(c\ln n)$-approximation algorithm, unless $P=NP$~\cite{DinurS14}.

Let $k$ be the size of a minimum dominating set of a graph $G$. We can assume w.l.o.g. that $k\geq 3$, as otherwise we can guess a minimum dominating set. We observe that $|D|\leq |D_k|$ and therefore we can show a contradiction on $D_k$ instead of $D$. In the following we focus on iteration $k$ of the for loop.

Let $\eta$ be the last iteration  of the while loop (the largest value of $j$ such that the while condition holds). 
We will show that, at each iteration of the while loop, $k$ nodes are added to $D_k$. Moreover, since the exit condition of the while loop is $|V_k^{\eta+1}| < k+1$, then the size of $V_k^{\eta+1}$ is at most $k$. Therefore the size of $D_k$ is at most $\eta k + k$, which implies that the approximation ratio of Algorithm~\ref{algo:harmonic:hardness} is at most $\eta+1$. In the following we show that $|S_k^j\cap V_k^j| = k$, for each $j\leq \eta$, and bound the value of $\eta$.

We first show that, for each $j\leq \eta$, any solution $S_k^j$ returned by algorithm $\mathcal{A}$ selects node $x$, i.e. $x\in S_k^j$. Indeed, we show that if $x\not\in S_k^j$, then we can find a node $u\in S_k^j$ such that $\GH(S_k^j) \leq ((S_k^j\cup\{x\})\setminus \{u\})$. We analyze three different cases.
\begin{itemize}
    \item If $y_i\in S_k^j$ for some $y_i\in Y_k^j$, then  $\GH(S_k^j) \leq ((S_k^j\cup\{x\})\setminus \{y_i\})$ since $\dist(S_k^j,x) = \dist((S_k^j\cup\{x\})\setminus \{y_i\},y_i)=1$ and any node different from $y_i$ is closer to $x$ than to $y_i$.
    \item If $Y_k^j\cap S_k^j = \emptyset$ and $z_i\in S_k^j$ for some $z_i\in Z_k^j$, then $\GH(S_k^j) \le 2+ \frac{k}{2} + h'$, where the first term is due to the two neighbors of $z_i$, the second term is due to the nodes in $Y$, and $h'$ is the contribution of any other node, note that all such nodes are at distance at least $2$ from $z_i$. By swapping $z_i$ with $x$ we obtain $\GH((S_k^j\cup\{x\})\setminus \{z_i\}) \ge 1 + k + h''$, where the first term is due to $z_i$,
    the second term is due to the nodes in $Y$, and $h''$ is the contribution of any other node, which are at distance at most $2$ from $x$, that is $h''\geq h'$. It follows that $\GH((S_k^j\cup\{x\})\setminus \{z_i\}) \geq \GH(S_k^j)$, for any $k\geq 2$.

   \item If $(Y_i\cup Z_k^j)\cap S_k^j = \emptyset$, that is $S_k^j\subseteq V_k^j$, then we show that there exists a node $v\in S_k^j$ such that $\GH(S_k^j) \leq ((S_k^j\cup\{x\})\setminus \{v\})$. For each $v\in S_k^j$, let us define $C(v):=\{w \in V_k^j~|~\dist(v,w)=1~\wedge~\dist(S_k^j\setminus \{v\},w)> 1 \}$, in other words, $C(v)$ are the nodes of $V_k^j$ that, among nodes in $S_k^j$, are adjacent only to $v$. Since, for $k\geq 1$, we have $|S_k^j|\geq 2$, then there exists at least a node $v\in S_k^j$ such that $|C(v)|\leq \lfloor n_k^j/2 \rfloor$. We observe that among nodes in $Z_k^j$ there are $k+1$ nodes at distance $1$ from $S_k^j$ and $n_k^j -k-1$ other nodes at distance at least $2$ (note that $n_k^j\geq k+1$ due to the condition of the while loop). The group harmonic centrality of $S_k^j$ is then $\GH(S_k^j) \leq \frac{k}{3} + \frac{1}{2} + k+1+ \frac{n_k^j -k-1}{2} + |C(v)| + h'$, where $ \frac{k}{3} + \frac{1}{2}$ is the contribution of nodes in $Y_k^j\cup\{x\}$ and  $h'$ is the contribution of nodes not in $Y_k^j\cup Z_k^j \cup\{x\}\cup C(v)$. By swapping $v$ with $x$ we obtain $\GH((S_k^j\cup\{x\})\setminus \{v\}) = k +n_k^j + \frac{|C(v)|}{2} + \frac{1}{2} + h''$, where the last term $\frac{1}{2}$ is due to $v$ and $h''$ is the contribution of nodes not in $Y_k^j\cup Z_k^j \cup\{x\}\cup C(v)$, with $h''\geq h'$. Since $|C(v)|\leq\lfloor n_k^j/2 \rfloor \leq n_k^j$ and $k\geq 3$, we obtain the statement.
\end{itemize}

We can further assume that, since $x\in S_k^j$, then $S_k^j$ does not contain any node in $Y_k^j\cup Z_k^j$. Indeed, if $y_i\in S_k^j$, for some $y_i\in Y_k^j$, then we can can swap $y_i$ with any node in $V_k^j\setminus S_k^j$. If $z_i\in S_k^j$, for some $z_i\in Z_k^j$, then we can swap $z_i$ with its neighbor in $V_k^j$, if it does not belong to $S_k^j$ or with any other node in $V_k^j\setminus S_k^j$ otherwise. In any case we do not decrease the value of the objective function.

Since $x\in S_k^j$ and $(Y_i\cup Z_k^j)\cap S_k^j = \emptyset$, it follows that $|S_k^j\cap V_k^j| = k$.

We now bound the value of $\eta$.
For each $j\leq \eta$, we have that the optimal value $OPT$ of the harmonic maximization problem on $G_k^j$ is at least $2n_k^j$. In fact, since $k$ is the size of an optimal dominating set of $G$, then there exists a dominating set of size $k$ for the subgraph of $G$ induced by $V_k^j$. If we select the $k$ nodes in a dominating set of this subgraph and node $x$, we have that all nodes in $V_k^j$ that are not selected and all nodes in $Y_k^j \cup Z_k^j$ are at distance 1 from the nodes in the solution. 

Let us consider the first iteration of the while loop (i.e. $j=1$) and let us denote as $c$ (as ``covered'') and $u$ (as ``uncovered'') the number of nodes in $V_k^1$ that are at distance $1$ and $2$ from $S_k^1$, respectively. Since $x\in S_k^1$ there is no node at distance greater than 2. We have that
\[
\GH(S_k^1) = c+ \frac{u}{2} + n_k^1 + k \geq \gamma 2n_k^1,
\]
since $S_k^1$ is a $\gamma$ approximation to $OPT$. 
Moreover, we have $n_k^1 = c+u+k$, that is $c = n_k^1 -u-k$, which implies 
\[
2n_k^1 -\frac{u}{2}\geq \gamma 2n_k^1,
\]
that is
\[
u \leq 4 n_k^1(1-\gamma).
\]
Note that $u$ is the number of nodes in $V$ that are given in input to the next iteration, i.e. $u=n_k^2$. By iterating the above arguments, we obtain
\[
n_k^j\leq 4n_k^{j-1}(1-\gamma) \leq n_k^1(4(1-\gamma))^{j-1},
\]
for each $j=2,\ldots,\eta$. By plugging $j=\eta$ and observing that $n_k^{\eta}\geq 1$, we obtain
\[
1\leq n_k^{\eta}\leq n_k^1(4(1-\gamma))^{\eta-1}.
\]
Since $\gamma>1-\frac{1}{4e}$, we have $4(1-\gamma)<1$, and hence the above inequality can be solved as
\[
\eta -1 \leq \log_{4(1-\gamma)}\frac{1}{n_k^1}=\frac{\ln(n_k^1)}{\ln((4(1-\gamma))^{-1})}.
\]

The approximation ratio of Algorithm~\ref{algo:harmonic:hardness} is at most $\eta +1 \leq \frac{\ln(n_k^1)}{\ln((4(1-\gamma))^{-1})} +2$. Let us denote $\alpha := \frac{1}{\ln((4(1-\gamma))^{-1})}$, since $\gamma>1-\frac{1}{4e}$, then $\alpha < 1$. For any $\beta$ such that $0<\alpha<\beta<1$ there exists a $n_\beta$ such that for each $n_k^1\geq n_\beta$, $\alpha\ln(n_k^1) + 2 \leq \beta\ln(n_k^1)$, which implies that the approximation ratio of Algorithm~\ref{algo:harmonic:hardness} is at most $\beta\ln(n_k^1)$. Since for any $c\in(0,1)$, there exist no $(c\ln n)$-approximation algorithm, unless $P=NP$~\cite{DinurS14}, we obtain a contradiction.

\section{Counter-example on the submodularity of group closeness}\label{apx:example}
We provide here a simple example illustrating that $\GC(\cdot)$ is not a submodular set 
function.\footnote{Note that another counter-example has already been pointed out in the most recent version of~\cite{BergaminiGM18}.}
Consider a simple path graph composed of four nodes $v_1$, $v_2$, $v_3$, and $v_4$. The edge-weight function $\ell$ is defined as follows: $\ell(\{v_1,v_2\})= L$ and $\ell(\{v_2,v_3\})=\ell(\{v_3,v_4\})=1$. It is easy to check that $\GC(\emptyset) = 0$, $\GC(\{v_1\}) = 4/(3L+3)$, $\GC(\{v_2\}) = 4/(L+3)$  and $\GC(\{v_1,v_2\}) = 4/3$. Hence, $\GC(\{v_1,v_2\}) - \GC(\{v_1\}) = 4L/(3(L+1))$ and $\GC(\{v_2\}) - \GC(\emptyset)  = 4/(L+3)$. It is straightforward that for a large enough value of $L$ (more precisely for $L\ge 2$), we have $\GC(\{v_1,v_2\}) - \GC(\{v_1\}) > \GC(\{v_2\}) - \GC(\emptyset)$ which shows that $\GC(\cdot)$ is not submodular.

\section{Approximation for group-closeness maximization in the Sense of Li et al.}\label{apx:li-approach}
The approach of Li et al.\ in fact works for minimizing a general supermodular monotone non-increasing function $f(\cdot)$ with respect to a cardinality constraint. They let $x_1^{*} \in \arg\max \{f(\emptyset) - f(\{x\})\}$ and use the greedy algorithm on the set function $g(S):= f(\{x^{*}_1\}) - f(\{x^{*}_1\}\cup S)$, which is a monotone non-decreasing submodular set function with $g(\emptyset) = 0$. Thus, the greedy algorithm maximizes the function with respect to a cardinality constraint within an approximation factor of $1-1/e$~\cite{NemhauserWF78}. However, there are two caveats. First, the greedy algorithm uses a budget of $k-1$ instead of $k$ (as a budget of one is spent on identifying $x_1^{*}$) and thus Li et al.\ obtain an approximation factor of $1-k/((k-1)e)$. Second and most importantly, observe that the approximation factor is obtained on the function $g(S)$ and not $f(S)$, i.e., they get a set $S$ of size $k-1$ such that 
$
    f(\{x_1^*\}) - f(S\cup \{x_1^*\})
    \ge \Big(1-\frac{k}{(k-1)\cdot e}\Big) \cdot
    (f(\{x_1^*\}) - f(S^*\cup \{x_1^*\}),
$ 
where $S^*$ is an optimal set of size $k-1$ for adding to $\{x_1^*$\} with the goal of minimizing $f$. We remark that this set is not necessarily related to the set that minimizes $f$ with respect to the cardinality constraint.
Clearly, this approach can be applied for the supermodular farness function $\GF(\cdot)$ in place of $f(\cdot)$. It can, however, not provide an approximation algorithm for $\GF(\cdot)$ in the usual sense --
and furthermore it would not be easily extendable to the closeness function $\GC(\cdot)$.

\input{ilp-formulation}
\input{additional-exp}
\vfill\eject
\input{instance-stats}
\input{pseudocodes}

\newpage
\input{running_times}

%% file: ilp-formulation.tex
\section{Ground Truth via ILP.}
To evaluate the quality of the results yielded by our greedy algorithm we
develop an ILP formulation of the group harmonic closeness maximization problem
similar to the one proposed in other
works~\cite{BergaminiGM18,crescenzi2016greedily} which we use later in our
experiments to compute the optimal solution $S^\star$ for some instances with
limited size and we compare it to the one yielded by our greedy algorithm.

We define a binary variable $y_j$ for each vertex $v_j \in V$ that is 1 if
$v_j \in S^\star$, 0 otherwise. A vertex $v_i$ is \emph{assigned} to $v_j\in
S^\star$ if $\dist(v_i, S^\star) = \dist(v_i, v_j)$ (if multiple vertices
satisfy this condition $v_i$ can be assigned arbitrarily to one of them).
For every node pair $(v_i, v_j)$ we define a variable $x_{ij}$ that is 1 if
$v_i$ is assigned to $v_j$, 0 otherwise.
Note that maximizing the sum of all $x_{ij}/\dist(v_i, v_j)$ would not work
because this would yield divisions by zero if a vertex is assigned to
itself. Thus, we set the contribution of all $x_{ii}$ to zero by splitting the
sum in two terms.

\begin{align*}
    \max&\sum_{i = 1}^n\left(
        \sum_{j = 1}^{i - 1} \frac{x_{ij}}{\dist(v_i, v_j)} +
        \sum_{j = i + 1}^{n}\frac{x_{ij}}{\dist(v_i, v_j)}
    \right)\label{eq:ghc-ilp}\\
        \text{s.t.}\, (i)\ & \sum_{j = 1}^n x_{ij} + y_i= 1 \quad \forall i \in \{1, \dots, n\}\\
        (ii)\ & \sum_{j = 1}^n y_j = k\\
        (iii)\ & x_{ij} \le y_j \quad \forall i,j \in \{1, \dots, n\}\\
         \text{where}\ & x_{ij}, y_j \in \{0, 1\}
\end{align*}

Condition $(i)$ states that each vertex but the ones in $S^\star$ is assigned
to exactly one vertex $v_j\in S^\star$, $(ii)$ that $|S^\star| = k$, and
$(iii)$ a vertex can be assigned only to vertices in $S^\star$.

%% file: additional-exp.tex
\section{Additional Experimental Results for Group Harmonic Maximization}
\label{apx:add-exp-harmonic}


\label{apx:exact-harmonic}
\FloatBarrier

\begin{minipage}[t]{\columnwidth}
\begin{figure}[H]
\centering
\begin{subfigure}[t]{\columnwidth}
\centering
\includegraphics{plots/legend-exact-harmonic}
\end{subfigure}\smallskip

\begin{subfigure}[t]{.5\columnwidth}
\centering
\includegraphics{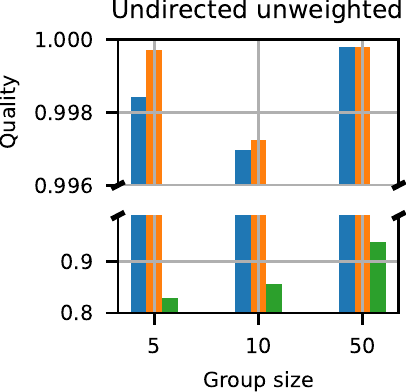}
\caption{Complex networks}
\end{subfigure}\hfill
\begin{subfigure}[t]{.5\columnwidth}
\centering
\includegraphics{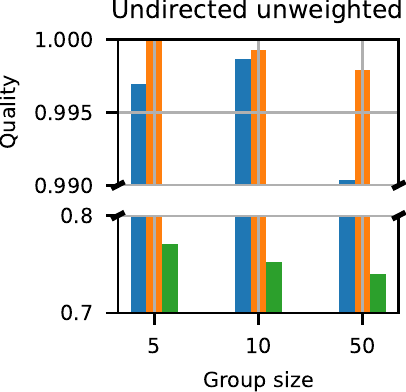}
\caption{High-diameter networks}
\end{subfigure}
\caption{Quality vs. the optimum over the small networks of Tables~\ref{tab:cplx-harmonic-small}
and~\ref{tab:high-diam-harmonic-small}.}
\label{fig:apx:quality-harmonic-ilp}
\end{figure}
\end{minipage}


\newpage
\label{apx:larger-harmonic}
\FloatBarrier

\begin{minipage}[t]{\columnwidth}
\begin{figure}[H]
\centering
\begin{subfigure}[t]{\columnwidth}
\centering
\includegraphics{plots/legend-large-harmonic}
\end{subfigure}\smallskip

\begin{subfigure}[t]{.5\columnwidth}
\centering
\includegraphics{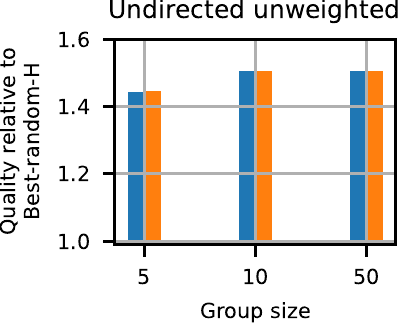}
\end{subfigure}\hfill
\begin{subfigure}[t]{.5\columnwidth}
\centering
\includegraphics{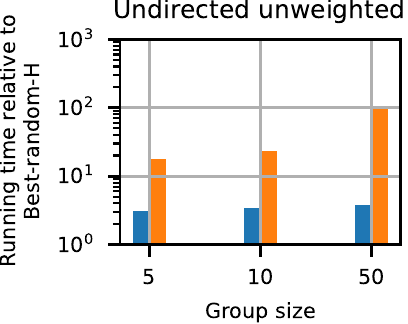}
\end{subfigure}

\caption{Quality and running time relative to \rndh over the large complex networks
of Table~\ref{tab:cplx-large}.}
\label{fig:apx:quality-harmonic-cplx}
\end{figure}

\begin{figure}[H]
\centering
\begin{subfigure}[t]{\columnwidth}
\centering
\includegraphics{plots/legend-large-harmonic}
\end{subfigure}\smallskip

\begin{subfigure}[t]{.5\columnwidth}
\centering
\includegraphics{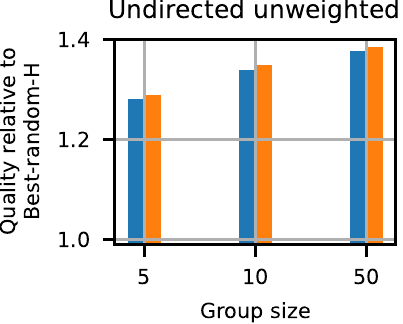}
\end{subfigure}\hfill
\begin{subfigure}[t]{.5\columnwidth}
\centering
\includegraphics{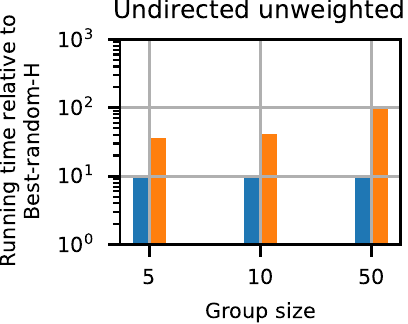}
\end{subfigure}\medskip

\begin{subfigure}[t]{.5\columnwidth}
\centering
\includegraphics{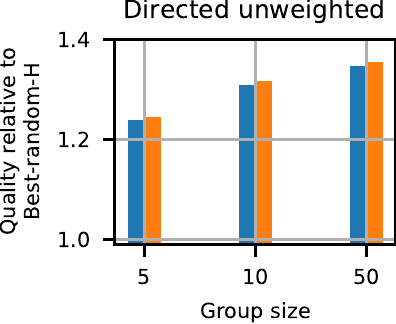}
\end{subfigure}\hfill
\begin{subfigure}[t]{.5\columnwidth}
\centering
\includegraphics{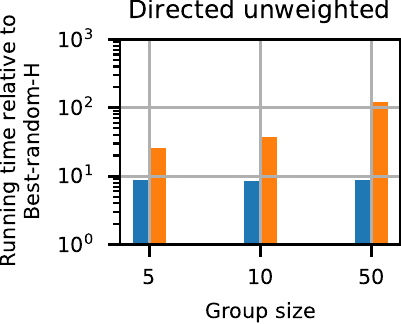}
\end{subfigure}\medskip

\begin{subfigure}[t]{.5\columnwidth}
\centering
\includegraphics{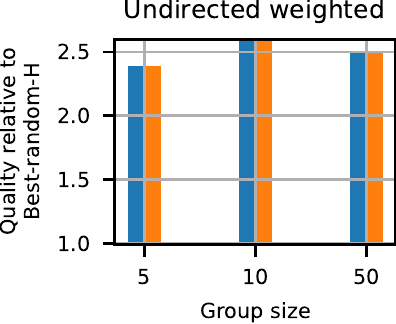}
\end{subfigure}\hfill
\begin{subfigure}[t]{.5\columnwidth}
\centering
\includegraphics{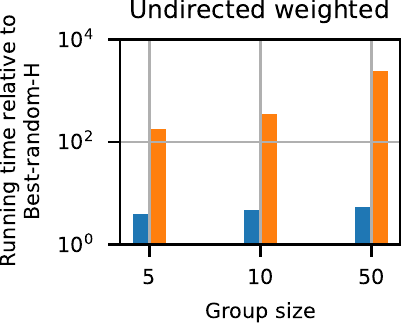}
\end{subfigure}\medskip

\caption{Quality and running time relative to \rndh over the large
high-diameter networks of Table~\ref{tab:high-diam-harmonic-large}.}
\label{fig:apx:quality-harmonic-high-diameter}
\end{figure}
\end{minipage}

\newpage


\section{Additional Experimental Results for Group Closeness Maximization}
\label{apx:add-exp-clos}

\begin{minipage}[t]{\columnwidth}
\begin{figure}[H]
\centering
\begin{subfigure}[t]{\columnwidth}
\centering
\includegraphics{plots/legend-ilp-closeness}
\end{subfigure}\smallskip

\begin{subfigure}[t]{.5\columnwidth}
\centering
\includegraphics{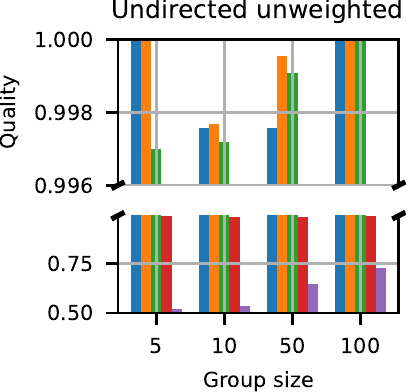}
\caption{Complex networks}
\end{subfigure}\hfill
\begin{subfigure}[t]{.5\columnwidth}
\centering
\includegraphics{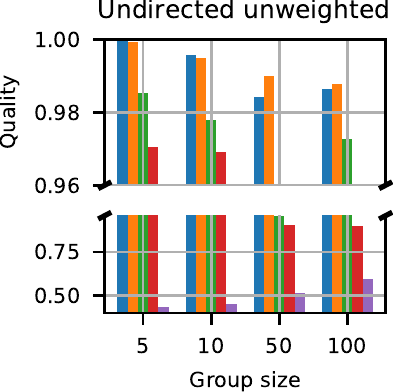}
\caption{High-diameter networks}
\end{subfigure}\medskip

\begin{subfigure}[t]{\columnwidth}
\begin{subfigure}[t]{.5\columnwidth}
\centering
\includegraphics{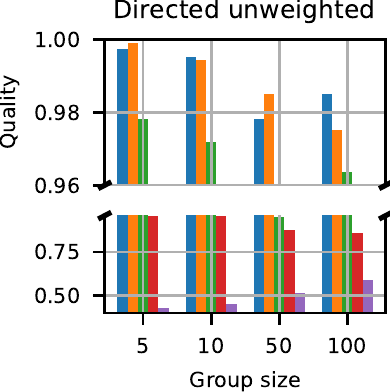}
\caption{High-diameter networks}
\end{subfigure}\hfill
\begin{subfigure}[t]{.5\columnwidth}
\centering
\includegraphics{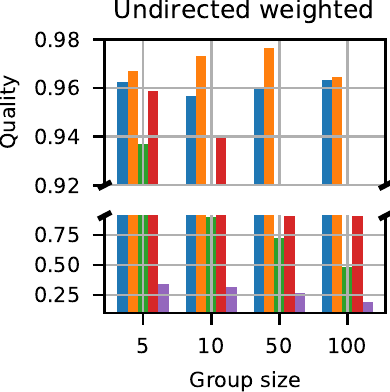}
\caption{High-diameter networks}
\end{subfigure}
\end{subfigure}

\caption{Quality vs the optimum over the small networks of Tables~\ref{tab:cplx-small}
and~\ref{tab:high-diam-small}.}
\label{fig:apx:quality-closeness-ilp}
\end{figure}
\end{minipage}

\vfill\eject


\begin{minipage}[t]{\columnwidth}
\begin{figure}[H]
\centering
\begin{subfigure}[t]{\columnwidth}
\centering
\includegraphics{plots/legend-quality-closeness}
\end{subfigure}\smallskip

\begin{subfigure}[t]{\columnwidth}
\begin{subfigure}[t]{.5\columnwidth}
\centering
\includegraphics{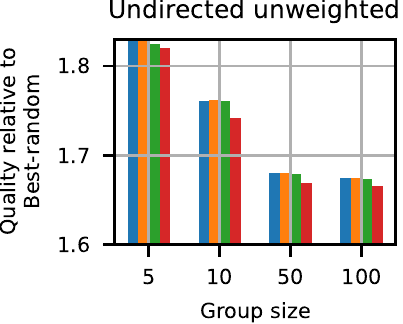}
\end{subfigure}\hfill
\begin{subfigure}[t]{.5\columnwidth}
\centering
\includegraphics{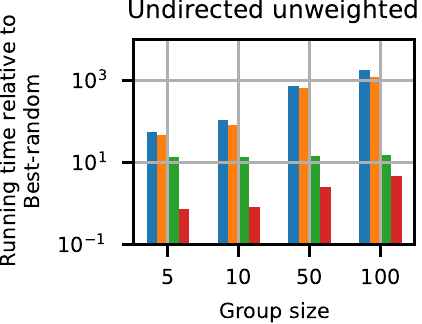}
\end{subfigure}
\caption{Complex networks}
\end{subfigure}\medskip

\begin{subfigure}[t]{\columnwidth}
\begin{subfigure}[t]{.5\columnwidth}
\centering
\includegraphics{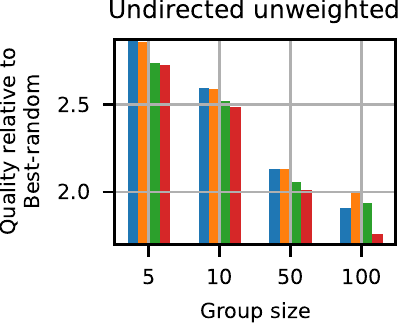}
\end{subfigure}\hfill
\begin{subfigure}[t]{.5\columnwidth}
\centering
\includegraphics{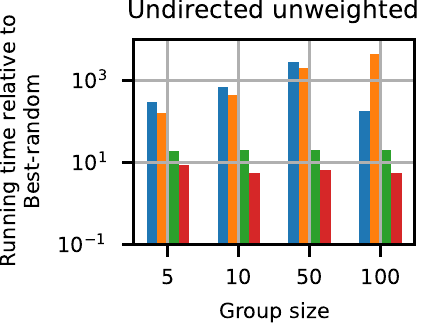}
\end{subfigure}
\caption{High-diameter diameter networks}
\end{subfigure}\medskip

\begin{subfigure}[t]{\columnwidth}
\begin{subfigure}[t]{.5\columnwidth}
\centering
\includegraphics{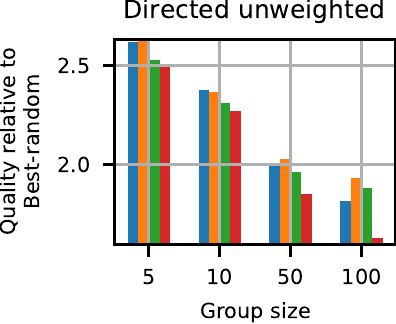}
\end{subfigure}\hfill
\begin{subfigure}[t]{.5\columnwidth}
\centering
\includegraphics{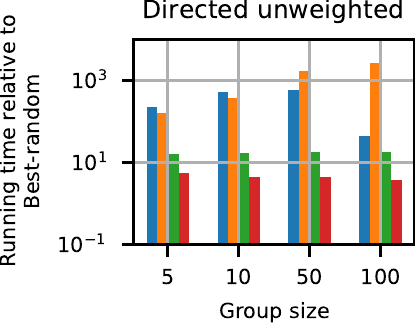}
\end{subfigure}
\caption{High-diameter networks}
\end{subfigure}\medskip

\begin{subfigure}[t]{\columnwidth}
\begin{subfigure}[t]{.5\columnwidth}
\centering
\includegraphics{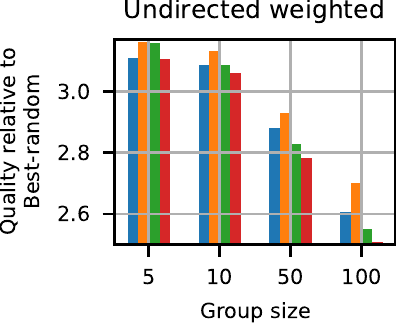}
\end{subfigure}\hfill
\begin{subfigure}[t]{.5\columnwidth}
\centering
\includegraphics{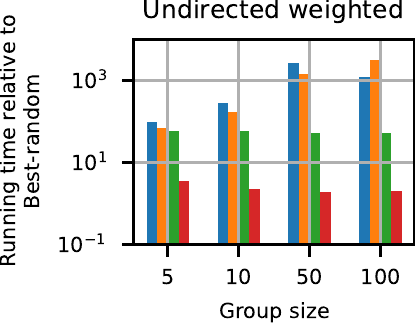}
\end{subfigure}
\caption{High-diameter networks}
\end{subfigure}

\caption{Quality and running time relative to \rnd over the complex networks of
Table~\ref{tab:cplx-large}.}
\label{fig:apx:quality-closeness}
\end{figure}
\end{minipage}

%% file: instance-stats.tex
\section{Instances Statistics}
\label{apx:insts-stats}
\FloatBarrier

\setlength{\tabcolsep}{2pt}




\begin{minipage}[tb]{\columnwidth}
\begin{table}[H]
\footnotesize
\centering
\input{tables/harmonic-closeness-small-diameter-small}
\caption{Small complex networks used for group harmonic closeness experiments with ILP solver.}
\label{tab:cplx-harmonic-small}
\end{table}
\end{minipage}

\begin{table}[tb]
\footnotesize
\centering
\input{tables/harmonic-closeness-high-diameter-small}
\caption{Small high-diameter networks used for group harmonic closeness experiments with ILP solver.}
\label{tab:high-diam-harmonic-small}
\end{table}




\begin{table}[tb]
\footnotesize
\centering
\input{tables/closeness-small-diameter-small}
\caption{Small complex networks used for group closeness experiments with ILP solver.}
\label{tab:cplx-small}
\end{table}

\begin{table}[tb]
\footnotesize
\centering
\input{tables/closeness-high-diameter-small}
\caption{Small high-diameter networks used for group closeness experiments with ILP solver.}
\label{tab:high-diam-small}
\end{table}

\begin{table}[tb]
\footnotesize
\centering
\input{tables/closeness-small-diameter-large}
\caption{Largest (strongly) connected components of the complex networks in
Table~\ref{tab:cplx-harmonic-large} used for group closeness experiments.}
\label{tab:cplx-large}
\end{table}

\begin{table}[tb]
\footnotesize
\centering
\input{tables/closeness-high-diameter-large}
\caption{Largest (strongly) connected components of the high-diameter networks in Table~\ref{tab:high-diam-harmonic-large}
used for group closeness experiments.}
\label{tab:high-diam-large}
\end{table}

%% file: tables/harmonic-closeness-small-diameter-small.tex
\begin{tabular}{lcrr}
\midrule
Graph & Type & $|V|$ & $|E|$\\
\midrule
convote & \texttt{U} & \numprint{219} & \numprint{586}\\
dimacs10-football & \texttt{U} & \numprint{115} & \numprint{613}\\
wiki\_talk\_ht & \texttt{U} & \numprint{537} & \numprint{787}\\
moreno\_innovation & \texttt{U} & \numprint{241} & \numprint{1098}\\
dimacs10-celegans\_metabolic & \texttt{U} & \numprint{453} & \numprint{2025}\\
arenas-meta & \texttt{U} & \numprint{453} & \numprint{2025}\\
foodweb-baywet & \texttt{U} & \numprint{128} & \numprint{2106}\\
contact & \texttt{U} & \numprint{275} & \numprint{2124}\\
foodweb-baydry & \texttt{U} & \numprint{128} & \numprint{2137}\\
moreno\_oz & \texttt{U} & \numprint{217} & \numprint{2672}\\
arenas-jazz & \texttt{U} & \numprint{198} & \numprint{2742}\\
sociopatterns-infectious & \texttt{U} & \numprint{411} & \numprint{2765}\\
dimacs10-celegansneural & \texttt{U} & \numprint{297} & \numprint{4296}\\
radoslaw\_email & \texttt{U} & \numprint{168} & \numprint{5783}\\
\midrule
convote & \texttt{D} & \numprint{219} & \numprint{586}\\
wiki\_talk\_ht & \texttt{D} & \numprint{537} & \numprint{787}\\
moreno\_innovation & \texttt{D} & \numprint{241} & \numprint{1098}\\
foodweb-baywet & \texttt{D} & \numprint{128} & \numprint{2106}\\
foodweb-baydry & \texttt{D} & \numprint{128} & \numprint{2137}\\
moreno\_oz & \texttt{D} & \numprint{217} & \numprint{2672}\\
dimacs10-celegansneural & \texttt{D} & \numprint{297} & \numprint{4296}\\
radoslaw\_email & \texttt{D} & \numprint{168} & \numprint{5783}\\
\midrule
\end{tabular}

%% file: tables/harmonic-closeness-high-diameter-small.tex
\begin{tabular}{lcrr}
\midrule
Graph & Type & $|V|$ & $|E|$\\
\midrule
dbpedia-similar & \texttt{UU} & \numprint{430} & \numprint{564}\\
niue & \texttt{UU} & \numprint{461} & \numprint{1055}\\
tuvalu & \texttt{UU} & \numprint{436} & \numprint{1082}\\
librec-filmtrust-trust & \texttt{UU} & \numprint{874} & \numprint{1853}\\
\midrule
niue & \texttt{UW} & \numprint{461} & \numprint{1055}\\
tuvalu & \texttt{UW} & \numprint{436} & \numprint{1082}\\
\midrule
niue & \texttt{DU} & \numprint{461} & \numprint{1055}\\
tuvalu & \texttt{DU} & \numprint{436} & \numprint{1082}\\
librec-filmtrust-trust & \texttt{DU} & \numprint{874} & \numprint{1853}\\
\midrule
niue & \texttt{DW} & \numprint{461} & \numprint{1055}\\
tuvalu & \texttt{DW} & \numprint{436} & \numprint{1082}\\
\midrule
\end{tabular}

%% file: tables/closeness-small-diameter-small.tex
\begin{tabular}{lcrr}
\midrule
Graph & Type & $|V|$ & $|E|$\\
\midrule
dimacs10-celegans\_metabolic & \texttt{U} & \numprint{453} & \numprint{2025}\\
arenas-meta & \texttt{U} & \numprint{453} & \numprint{2025}\\
contact & \texttt{U} & \numprint{274} & \numprint{2124}\\
arenas-jazz & \texttt{U} & \numprint{198} & \numprint{2742}\\
sociopatterns-infectious & \texttt{U} & \numprint{410} & \numprint{2765}\\
dnc-corecipient & \texttt{U} & \numprint{849} & \numprint{10384}\\
\midrule
moreno\_oz & \texttt{D} & \numprint{214} & \numprint{2658}\\
wiki\_talk\_lv & \texttt{D} & \numprint{510} & \numprint{2783}\\
wiki\_talk\_eu & \texttt{D} & \numprint{617} & \numprint{2811}\\
dnc-temporalGraph & \texttt{D} & \numprint{520} & \numprint{3518}\\
dimacs10-celegansneural & \texttt{D} & \numprint{297} & \numprint{4296}\\
wiki\_talk\_bn & \texttt{D} & \numprint{700} & \numprint{4316}\\
wiki\_talk\_eo & \texttt{D} & \numprint{822} & \numprint{6076}\\
wiki\_talk\_gl & \texttt{D} & \numprint{1009} & \numprint{7435}\\
\midrule
\end{tabular}

%% file: tables/closeness-high-diameter-small.tex
\begin{tabular}{lcrr}
\midrule
Graph & Type & $|V|$ & $|E|$\\
\midrule
tuvalu & \texttt{UU} & \numprint{152} & \numprint{187}\\
niue & \texttt{UU} & \numprint{461} & \numprint{529}\\
nauru & \texttt{UU} & \numprint{618} & \numprint{729}\\
dimacs10-netscience & \texttt{UU} & \numprint{379} & \numprint{914}\\
asoiaf & \texttt{UU} & \numprint{796} & \numprint{2823}\\
\midrule
tuvalu & \texttt{UW} & \numprint{152} & \numprint{187}\\
niue & \texttt{UW} & \numprint{461} & \numprint{529}\\
nauru & \texttt{UW} & \numprint{618} & \numprint{729}\\
\midrule
tuvalu & \texttt{DU} & \numprint{152} & \numprint{374}\\
niue & \texttt{DU} & \numprint{461} & \numprint{1055}\\
librec-filmtrust-trust & \texttt{DU} & \numprint{267} & \numprint{1099}\\
nauru & \texttt{DU} & \numprint{618} & \numprint{1427}\\
\midrule
tuvalu & \texttt{DW} & \numprint{152} & \numprint{374}\\
niue & \texttt{DW} & \numprint{461} & \numprint{1055}\\
nauru & \texttt{DW} & \numprint{618} & \numprint{1427}\\
\midrule
\end{tabular}

%% file: tables/closeness-small-diameter-large.tex
\begin{tabular}{lcrr}
\midrule
Graph & Type & $|V|$ & $|E|$\\
\midrule
loc-brightkite\_edges & \texttt{U} & \numprint{56739} & \numprint{212945}\\
douban & \texttt{U} & \numprint{154908} & \numprint{327162}\\
petster-cat-household & \texttt{U} & \numprint{68315} & \numprint{494562}\\
wikipedia\_link\_ckb & \texttt{U} & \numprint{60257} & \numprint{801794}\\
wikipedia\_link\_fy & \texttt{U} & \numprint{65512} & \numprint{921533}\\
livemocha & \texttt{U} & \numprint{104103} & \numprint{2193083}\\
\midrule
wikipedia\_link\_mi & \texttt{D} & \numprint{3696} & \numprint{99237}\\
wikipedia\_link\_lo & \texttt{D} & \numprint{1622} & \numprint{109577}\\
wikipedia\_link\_so & \texttt{D} & \numprint{5149} & \numprint{114922}\\
foldoc & \texttt{D} & \numprint{13274} & \numprint{119485}\\
wikipedia\_link\_co & \texttt{D} & \numprint{5150} & \numprint{160474}\\
web-NotreDame & \texttt{D} & \numprint{53968} & \numprint{296228}\\
slashdot-zoo & \texttt{D} & \numprint{26997} & \numprint{333425}\\
soc-Epinions1 & \texttt{D} & \numprint{32223} & \numprint{443506}\\
wikipedia\_link\_jv & \texttt{D} & \numprint{39248} & \numprint{1059059}\\
\midrule
\end{tabular}

%% file: tables/closeness-high-diameter-large.tex
\begin{tabular}{lcrr}
\midrule
Graph & Type & $|V|$ & $|E|$\\
\midrule
seychelles & \texttt{UU} & \numprint{3907} & \numprint{4322}\\
comores & \texttt{UU} & \numprint{3789} & \numprint{4630}\\
andorra & \texttt{UU} & \numprint{4219} & \numprint{4933}\\
opsahl-powergrid & \texttt{UU} & \numprint{4941} & \numprint{6594}\\
liechtenstein & \texttt{UU} & \numprint{6215} & \numprint{7002}\\
faroe-islands & \texttt{UU} & \numprint{12129} & \numprint{13165}\\
\midrule
seychelles & \texttt{UW} & \numprint{3907} & \numprint{4322}\\
comores & \texttt{UW} & \numprint{3789} & \numprint{4630}\\
andorra & \texttt{UW} & \numprint{4219} & \numprint{4933}\\
liechtenstein & \texttt{UW} & \numprint{6215} & \numprint{7002}\\
faroe-islands & \texttt{UW} & \numprint{12129} & \numprint{13165}\\
DC & \texttt{UW} & \numprint{9522} & \numprint{14807}\\
\midrule
seychelles & \texttt{DU} & \numprint{3907} & \numprint{8225}\\
andorra & \texttt{DU} & \numprint{4160} & \numprint{8288}\\
comores & \texttt{DU} & \numprint{3789} & \numprint{8952}\\
liechtenstein & \texttt{DU} & \numprint{6205} & \numprint{13591}\\
faroe-islands & \texttt{DU} & \numprint{12077} & \numprint{25679}\\
opsahl-openflights & \texttt{DU} & \numprint{2868} & \numprint{30404}\\
tntp-ChicagoRegional & \texttt{DU} & \numprint{12978} & \numprint{39017}\\
\midrule
seychelles & \texttt{DW} & \numprint{3907} & \numprint{8225}\\
andorra & \texttt{DW} & \numprint{4160} & \numprint{8288}\\
comores & \texttt{DW} & \numprint{3789} & \numprint{8952}\\
liechtenstein & \texttt{DW} & \numprint{6205} & \numprint{13591}\\
faroe-islands & \texttt{DW} & \numprint{12077} & \numprint{25679}\\
\midrule
\end{tabular}

%% file: pseudocodes.tex
\FloatBarrier
\section{Pseudocodes}
\label{apx:pseudocodes}
\vspace{-15pt}
\input{group-harmonic-greedy}
\input{single-swaps-algo}

%% file: group-harmonic-greedy.tex
\begin{minipage}[t]{\columnwidth}
\begin{algorithm}[H]
\caption{Greedy algorithm for group-harmonic closeness}
\begin{algorithmic}[1]
\State$v \gets \texttt{topHarmonicCloseness}()$;\ \ $S \gets \{v\}$ \label{line:top-harmonic-vtx}
\While{$|S| < k$}\label{line:group-hclos-while}
    \State$PQ \gets $ max-PQ with key $\GHhat(S, u)$ and value $u$\label{line:group-hclos-pq}
    \For{each $u\in V \setminus S$}\label{line:max-pq-insert1}
        \State$PQ.\texttt{push}(u)$
    \EndFor\label{line:max-pq-insert2}
    \State$x \gets \nil$
    \State $\GH(S \cup \{x\})$ $\gets -\infty$
    \Repeat\Comment{This loop is done in parallel.}\label{line:group-hclos-repeat}
        \State $u \gets PQ.\texttt{extract\_max}()$
        \If{$\GHhat(S, u) \le \GH(S \cup \{x\})$}
        \State\textbf{break}\Comment{$x$ has the highest marginal gain.}
        \EndIf
        \State (isExact, $\GH(S \cup \{u\})) \gets $ pruned SSSP($u, \GH(S \cup \{x\})$)\label{line:group-hclos-sssp}
        \If{isExact \textbf{and} $\GH(S \cup \{u\}) > \GH(S \cup \{x\})$}
            \State $x \gets u$
        \EndIf
    \Until{$PQ$ is empty}
    \State$S \gets S \cup \{x\}$
\EndWhile
\State\Return$S$
\end{algorithmic}
\label{algo:greedy-group-harmonic}
\end{algorithm}
\end{minipage}

%% file: single-swaps-algo.tex
\begin{minipage}[t]{\columnwidth}
\begin{algorithm}[H]
\caption{Overview of the single-swap algorithm}
\label{algo:local-search}
\begin{algorithmic}[1]
\State $S\gets$ \textit{grow-shrink}$(G, k)$
\State $\GF(S) \gets SSSP(S)$
\Repeat
    \State $PQ_u \gets$ min-PQ with key $(\GF(S \setminus \{u\}) - \GF(S))$ and value $u$\label{line:pq1}
    \For{each $w\in S$}
    \State $PQ_u.\texttt{push}(w)$
    \EndFor\label{line:pq2}
    \State didSwap $\gets$ \false
    \Repeat\label{line:outer-repeat}
        \State $u \gets PQ_u.\texttt{extract\_min}()$\\
        \Comment{Compute exact farness increase}
        \State $\GF^+(u)$ $\gets \GF(S \setminus \{u\}) - \GF(S)$\label{line:farn-inc}
        \State compute $\GFapx(\Suv)$ for all $V \setminus S$\label{line:maxpq1}
        \State $PQ_v \gets$ max-PQ with key $\GFapx(\Suv)$ and value $v$
        \For{each $w\in V\setminus S$}
        \State $PQ_v.\texttt{push}(w)$
        \EndFor\label{line:maxpq2}
        \Repeat\label{line:inner-repeat}\Comment{This loop is done in parallel.}
        \State \label{line:pick-swap} $v \gets PQ.\texttt{extract\_max()}$
        \State $\GF(\Suv)\gets$ pruned SSSP from $v$\label{line:pruned-sssp}\Comment{Compute exact farness decrement.}
        \If{$\GF(\Suv) \le (1 - \frac{\varepsilon}{k\cdot(n - k)}) \GF(S)$}
            \State $S \gets \Suv$
            \State $\GF(S) \gets SSSP(S)$
            \State didSwap $\gets$ \true
            \State\textbf{break}
        \EndIf
        \Until{$PQ_v$ is empty}
    \If{didSwap}
    \State\textbf{break}
    \EndIf
    \Until{$PQ_u$ is empty}
\Until{\texttt{not} didSwap}
\State\Return$S$
\end{algorithmic}
\end{algorithm}
\end{minipage}

%% file: running_times.tex
\section{Running Times}
\label{sec:running-times}

\begin{minipage}[t]{\columnwidth}
\begin{table}[H]
\footnotesize
\centering
\tabtitle{Undirected unweighted}

\input{./tables/time-h-undirected-unweighted-small-diameter}

\tabtitle{Directed unweighted}

\input{./tables/time-h-directed-unweighted-small-diameter}
\caption{Running time (s) of \greedyh and \grlsh
on the complex networks of Table~\ref{tab:cplx-harmonic-large}.}
\label{tab:runtime-h-cplx}
\end{table}
\end{minipage}

\begin{minipage}[t]{\columnwidth}
\begin{table}[H]
\footnotesize
\centering
\tabtitle{Undirected unweighted}

\input{./tables/time-h-undirected-unweighted-high-diameter}
\medskip

\tabtitle{Undirected weighted}

\input{./tables/time-h-undirected-weighted-high-diameter}
\medskip

\tabtitle{Directed unweighted}

\input{./tables/time-h-directed-unweighted-high-diameter}
\medskip

\tabtitle{Directed weighted}

\input{./tables/time-h-directed-weighted-high-diameter}

\caption{Running time (s) of \greedyh and \grlsh
on the high-diameter networks of Table~\ref{tab:high-diam-harmonic-large}.}
\label{tab:runtime-h-high-diam}
\end{table}
\end{minipage}

\begin{minipage}[t]{\columnwidth}
\begin{table}[H]
\footnotesize
\centering
\tabtitle{Undirected unweighted}

\input{tables/time-c-undirected-unweighted-small-diameter}\medskip

\tabtitle{Directed unweighted}

\input{tables/time-c-directed-unweighted-small-diameter}

\caption{Running time (s) of \gsls and \grls on the complex networks
of Table~\ref{tab:cplx-large}.}
\label{tab:runtime-c-cplx}
\end{table}
\end{minipage}

\begin{minipage}[t]{\columnwidth}
\begin{table}[H]
\footnotesize
\centering
\tabtitle{Undirected unweighted}

\input{tables/time-c-undirected-unweighted-high-diameter}
\medskip

\tabtitle{Undirected weighted}

\input{tables/time-c-undirected-weighted-high-diameter}
\medskip

\tabtitle{Directed unweighted}

\input{tables/time-c-directed-unweighted-high-diameter}
\medskip

\tabtitle{Directed weighted}

\input{tables/time-c-directed-weighted-high-diameter}

\caption{Running time (s) of \gsls and \grls on the
high-diameter networks of Table~\ref{tab:high-diam-large}.}
\label{tab:runtime-c-high-diam}
\end{table}
\end{minipage}

%% file: tables/time-h-undirected-unweighted-small-diameter.tex
\begin{tabular}{lrrrrrr}
\toprule
Graph & \multicolumn{3}{c}{Greedy-H} & \multicolumn{3}{c}{Greedy-LS-H}\\
\hfill $k$ & $5$ & $10$ & $50$ & $5$ & $10$ & $50$\\
\midrule
petster-hamster-household & \textless\numprint{0.1} & \textless\numprint{0.1} & \textless\numprint{0.1} & \textless\numprint{0.1} & \textless\numprint{0.1} & \textless\numprint{0.1}\\
petster-hamster-friend & \textless\numprint{0.1} & \textless\numprint{0.1} & \textless\numprint{0.1} & \textless\numprint{0.1} & \textless\numprint{0.1} & \numprint{0.1}\\
petster-hamster & \textless\numprint{0.1} & \textless\numprint{0.1} & \textless\numprint{0.1} & \textless\numprint{0.1} & \textless\numprint{0.1} & \textless\numprint{0.1}\\
loc-brightkite\_edges & \numprint{1.1} & \numprint{1.0} & \numprint{1.1} & \numprint{4.3} & \numprint{6.6} & \numprint{25.8}\\
douban & \numprint{8.1} & \numprint{8.1} & \numprint{8.4} & \numprint{40.3} & \numprint{86.3} & \numprint{303.0}\\
petster-cat-household & \numprint{0.1} & \numprint{0.2} & \numprint{0.3} & \numprint{19.5} & \numprint{23.8} & \numprint{106.1}\\
loc-gowalla\_edges & \numprint{8.9} & \numprint{8.4} & \numprint{8.7} & \numprint{59.8} & \numprint{97.3} & \numprint{1064.5}\\
wikipedia\_link\_fy & \numprint{3.8} & \numprint{3.8} & \numprint{4.0} & \numprint{13.3} & \numprint{15.7} & \numprint{137.9}\\
wikipedia\_link\_ckb & \numprint{7.3} & \numprint{7.3} & \numprint{7.4} & \numprint{12.9} & \numprint{14.6} & \numprint{80.2}\\
petster-dog-household & \numprint{10.3} & \numprint{10.4} & \numprint{10.7} & \numprint{131.9} & \numprint{212.3} & \numprint{843.8}\\
livemocha & \numprint{11.2} & \numprint{11.4} & \numprint{11.8} & \numprint{52.5} & \numprint{64.6} & \numprint{277.9}\\
flickrEdges & \numprint{44.2} & \numprint{45.4} & \numprint{46.4} & \numprint{119.5} & \numprint{128.4} & \numprint{217.6}\\
petster-friendships-cat & \numprint{2.7} & \numprint{2.8} & \numprint{2.9} & \numprint{35.6} & \numprint{55.1} & \numprint{266.7}\\
\bottomrule
\end{tabular}

%% file: tables/time-h-directed-unweighted-small-diameter.tex
\begin{tabular}{lrrrrrr}
\toprule
Graph & \multicolumn{3}{c}{Greedy-H} & \multicolumn{3}{c}{Greedy-LS-H}\\
\hfill $k$ & $5$ & $10$ & $50$ & $5$ & $10$ & $50$\\
\midrule
wikipedia\_link\_mi & \numprint{0.3} & \numprint{0.3} & \numprint{0.3} & \numprint{0.6} & \numprint{1.0} & \numprint{3.8}\\
foldoc & \numprint{0.6} & \numprint{0.6} & \numprint{0.6} & \numprint{1.6} & \numprint{1.7} & \numprint{14.7}\\
wikipedia\_link\_so & \numprint{0.1} & \numprint{0.1} & \numprint{0.1} & \numprint{0.3} & \numprint{0.4} & \numprint{2.2}\\
wikipedia\_link\_lo & \numprint{0.2} & \numprint{0.2} & \numprint{0.2} & \numprint{0.2} & \numprint{0.2} & \numprint{0.7}\\
wikipedia\_link\_co & \numprint{0.2} & \numprint{0.3} & \numprint{0.3} & \numprint{0.5} & \numprint{0.5} & \numprint{2.6}\\
\bottomrule
\end{tabular}

%% file: tables/time-h-undirected-unweighted-high-diameter.tex
\begin{tabular}{lrrrrrr}
\toprule
Graph & \multicolumn{3}{c}{Greedy-H} & \multicolumn{3}{c}{Greedy-LS-H}\\
\hfill $k$ & $5$ & $10$ & $50$ & $5$ & $10$ & $50$\\
\midrule
marshall-islands & \textless\numprint{0.1} & \textless\numprint{0.1} & \textless\numprint{0.1} & \textless\numprint{0.1} & \textless\numprint{0.1} & \textless\numprint{0.1}\\
micronesia & \textless\numprint{0.1} & \textless\numprint{0.1} & \textless\numprint{0.1} & \textless\numprint{0.1} & \textless\numprint{0.1} & \numprint{0.2}\\
kiribati & \textless\numprint{0.1} & \textless\numprint{0.1} & \textless\numprint{0.1} & \textless\numprint{0.1} & \textless\numprint{0.1} & \numprint{0.3}\\
opsahl-powergrid & \numprint{0.2} & \numprint{0.2} & \numprint{0.2} & \numprint{1.7} & \numprint{0.9} & \numprint{1.4}\\
samoa & \numprint{0.8} & \numprint{0.9} & \numprint{0.9} & \numprint{2.6} & \numprint{2.9} & \numprint{5.5}\\
comores & \numprint{0.5} & \numprint{0.5} & \numprint{0.6} & \numprint{1.1} & \numprint{2.3} & \numprint{8.6}\\
\bottomrule
\end{tabular}

%% file: tables/time-h-undirected-weighted-high-diameter.tex
\begin{tabular}{lrrrrrr}
\toprule
Graph & \multicolumn{3}{c}{Greedy-H} & \multicolumn{3}{c}{Greedy-LS-H}\\
\hfill $k$ & $5$ & $10$ & $50$ & $5$ & $10$ & $50$\\
\midrule
marshall-islands & \textless\numprint{0.1} & \textless\numprint{0.1} & \textless\numprint{0.1} & \numprint{0.3} & \numprint{0.8} & \numprint{6.1}\\
micronesia & \textless\numprint{0.1} & \textless\numprint{0.1} & \textless\numprint{0.1} & \numprint{1.0} & \numprint{2.6} & \numprint{22.1}\\
kiribati & \textless\numprint{0.1} & \textless\numprint{0.1} & \textless\numprint{0.1} & \numprint{1.1} & \numprint{2.3} & \numprint{21.9}\\
DC & \numprint{4.7} & \numprint{4.8} & \numprint{4.9} & \numprint{86.3} & \numprint{161.0} & \numprint{2247.7}\\
samoa & \numprint{1.3} & \numprint{1.9} & \numprint{2.4} & \numprint{30.0} & \numprint{60.8} & \numprint{323.3}\\
comores & \numprint{0.2} & \numprint{0.4} & \numprint{0.8} & \numprint{26.6} & \numprint{64.3} & \numprint{732.6}\\
\bottomrule
\end{tabular}

%% file: tables/time-h-directed-unweighted-high-diameter.tex
\begin{tabular}{lrrrrrr}
\toprule
Graph & \multicolumn{3}{c}{Greedy-H} & \multicolumn{3}{c}{Greedy-LS-H}\\
\hfill $k$ & $5$ & $10$ & $50$ & $5$ & $10$ & $50$\\
\midrule
marshall-islands & \textless\numprint{0.1} & \textless\numprint{0.1} & \textless\numprint{0.1} & \textless\numprint{0.1} & \textless\numprint{0.1} & \textless\numprint{0.1}\\
micronesia & \textless\numprint{0.1} & \textless\numprint{0.1} & \textless\numprint{0.1} & \textless\numprint{0.1} & \textless\numprint{0.1} & \numprint{0.2}\\
kiribati & \textless\numprint{0.1} & \textless\numprint{0.1} & \textless\numprint{0.1} & \textless\numprint{0.1} & \textless\numprint{0.1} & \numprint{0.3}\\
samoa & \numprint{0.8} & \numprint{0.9} & \numprint{0.9} & \numprint{2.6} & \numprint{1.9} & \numprint{5.9}\\
comores & \numprint{0.5} & \numprint{0.5} & \numprint{0.6} & \numprint{1.1} & \numprint{3.3} & \numprint{10.5}\\
opsahl-openflights & \textless\numprint{0.1} & \textless\numprint{0.1} & \textless\numprint{0.1} & \textless\numprint{0.1} & \textless\numprint{0.1} & \numprint{0.2}\\
tntp-ChicagoRegional & \numprint{2.7} & \numprint{2.9} & \numprint{3.2} & \numprint{11.5} & \numprint{20.6} & \numprint{85.8}\\
\bottomrule
\end{tabular}

%% file: tables/time-h-directed-weighted-high-diameter.tex
\begin{tabular}{lrrrrrr}
\toprule
Graph & \multicolumn{3}{c}{Greedy-H} & \multicolumn{3}{c}{Greedy-LS-H}\\
\hfill $k$ & $5$ & $10$ & $50$ & $5$ & $10$ & $50$\\
\midrule
marshall-islands & \textless\numprint{0.1} & \textless\numprint{0.1} & \textless\numprint{0.1} & \numprint{0.3} & \numprint{0.8} & \numprint{6.1}\\
micronesia & \textless\numprint{0.1} & \textless\numprint{0.1} & \textless\numprint{0.1} & \numprint{1.0} & \numprint{2.6} & \numprint{22.2}\\
kiribati & \textless\numprint{0.1} & \textless\numprint{0.1} & \textless\numprint{0.1} & \numprint{1.1} & \numprint{2.3} & \numprint{22.0}\\
samoa & \numprint{1.2} & \numprint{1.8} & \numprint{2.3} & \numprint{36.8} & \numprint{63.4} & \numprint{331.5}\\
comores & \numprint{0.3} & \numprint{0.4} & \numprint{0.8} & \numprint{27.4} & \numprint{68.0} & \numprint{521.6}\\
\bottomrule
\end{tabular}

%% file: tables/time-c-undirected-unweighted-small-diameter.tex
\begin{tabular}{lrrrrrr}
\toprule
Graph & \multicolumn{3}{c}{GS-LS-C} & \multicolumn{3}{c}{Greedy-LS-C}\\
\hfill $k$ & $5$ & $10$ & $50$ & $5$ & $10$ & $50$\\
\midrule
loc-brightkite\_edges & \numprint{11.8} & \numprint{22.1} & \numprint{146.4} & \numprint{11.5} & \numprint{20.8} & \numprint{110.9}\\
douban & \numprint{35.0} & \numprint{59.4} & \numprint{222.0} & \numprint{26.3} & \numprint{43.5} & \numprint{202.5}\\
petster-cat-household & \numprint{32.7} & \numprint{66.1} & \numprint{363.2} & \numprint{32.2} & \numprint{63.2} & \numprint{341.5}\\
wikipedia\_link\_fy & \numprint{100.2} & \numprint{102.5} & \numprint{476.9} & \numprint{27.5} & \numprint{50.3} & \numprint{434.7}\\
wikipedia\_link\_ckb & \numprint{19.7} & \numprint{103.2} & \numprint{767.2} & \numprint{19.4} & \numprint{34.1} & \numprint{718.3}\\
livemocha & \numprint{58.1} & \numprint{86.3} & \numprint{713.0} & \numprint{46.5} & \numprint{58.2} & \numprint{604.9}\\
\bottomrule
\end{tabular}

%% file: tables/time-c-directed-unweighted-small-diameter.tex
\begin{tabular}{lrrrrrr}
\toprule
Graph & \multicolumn{3}{c}{GS-LS-C} & \multicolumn{3}{c}{Greedy-LS-C}\\
\hfill $k$ & $5$ & $10$ & $50$ & $5$ & $10$ & $50$\\
\midrule
wikipedia\_link\_mi & \textless\numprint{0.1} & \numprint{0.8} & \numprint{2.9} & \numprint{0.1} & \numprint{0.2} & \numprint{1.8}\\
foldoc & \numprint{2.3} & \numprint{3.5} & \numprint{0.5} & \numprint{1.7} & \numprint{2.2} & \numprint{5.7}\\
wikipedia\_link\_so & \numprint{0.7} & \numprint{1.5} & \numprint{23.7} & \numprint{0.5} & \numprint{0.9} & \numprint{3.3}\\
wikipedia\_link\_lo & \numprint{0.8} & \numprint{1.6} & \numprint{26.0} & \numprint{0.5} & \numprint{2.1} & \numprint{13.5}\\
wikipedia\_link\_co & \numprint{0.8} & \numprint{1.7} & \numprint{42.9} & \numprint{1.2} & \numprint{1.7} & \numprint{18.5}\\
soc-Epinions1 & \numprint{4.2} & \numprint{6.9} & \numprint{30.1} & \numprint{3.6} & \numprint{6.0} & \numprint{28.2}\\
slashdot-zoo & \numprint{4.1} & \numprint{6.4} & \numprint{19.9} & \numprint{3.4} & \numprint{7.1} & \numprint{15.4}\\
web-NotreDame & \numprint{14.9} & \numprint{37.3} & \numprint{1106.5} & \numprint{14.4} & \numprint{23.4} & \numprint{388.6}\\
wikipedia\_link\_jv & \numprint{22.9} & \numprint{97.1} & \numprint{30.0} & \numprint{17.7} & \numprint{14.5} & \numprint{49.9}\\
\bottomrule
\end{tabular}

%% file: tables/time-c-undirected-unweighted-high-diameter.tex
\begin{tabular}{lrrrrrr}
\toprule
Graph & \multicolumn{3}{c}{GS-LS-C} & \multicolumn{3}{c}{Greedy-LS-C}\\
\hfill $k$ & $5$ & $10$ & $50$ & $5$ & $10$ & $50$\\
\midrule
opsahl-powergrid & \numprint{0.9} & \numprint{1.1} & \numprint{13.4} & \numprint{0.7} & \numprint{0.4} & \numprint{3.7}\\
andorra & \numprint{3.6} & \numprint{8.9} & \numprint{55.1} & \numprint{1.9} & \numprint{3.9} & \numprint{26.7}\\
seychelles & \numprint{1.6} & \numprint{5.3} & \numprint{26.7} & \numprint{0.9} & \numprint{3.4} & \numprint{25.0}\\
liechtenstein & \numprint{10.8} & \numprint{21.2} & \numprint{56.3} & \numprint{2.2} & \numprint{16.4} & \numprint{38.2}\\
comores & \numprint{1.2} & \numprint{5.0} & \numprint{22.9} & \numprint{1.4} & \numprint{4.9} & \numprint{18.0}\\
faroe-islands & \numprint{33.9} & \numprint{77.2} & \numprint{313.5} & \numprint{25.3} & \numprint{96.4} & \numprint{268.4}\\
\bottomrule
\end{tabular}

%% file: tables/time-c-undirected-weighted-high-diameter.tex
\begin{tabular}{lrrrrrr}
\toprule
Graph & \multicolumn{3}{c}{GS-LS-C} & \multicolumn{3}{c}{Greedy-LS-C}\\
\hfill $k$ & $5$ & $10$ & $50$ & $5$ & $10$ & $50$\\
\midrule
andorra & \numprint{20.5} & \numprint{35.5} & \numprint{182.0} & \numprint{4.5} & \numprint{10.9} & \numprint{64.3}\\
seychelles & \numprint{2.6} & \numprint{13.7} & \numprint{93.1} & \numprint{2.3} & \numprint{3.3} & \numprint{62.6}\\
liechtenstein & \numprint{3.8} & \numprint{8.6} & \numprint{230.3} & \numprint{4.1} & \numprint{27.0} & \numprint{265.6}\\
DC & \numprint{7.8} & \numprint{18.9} & \numprint{473.2} & \numprint{9.9} & \numprint{14.0} & \numprint{98.3}\\
comores & \numprint{2.3} & \numprint{10.1} & \numprint{140.1} & \numprint{2.3} & \numprint{9.6} & \numprint{55.3}\\
faroe-islands & \numprint{17.1} & \numprint{137.5} & \numprint{907.0} & \numprint{15.6} & \numprint{27.4} & \numprint{411.3}\\
\bottomrule
\end{tabular}

%% file: tables/time-c-directed-unweighted-high-diameter.tex
\begin{tabular}{lrrrrrr}
\toprule
Graph & \multicolumn{3}{c}{GS-LS-C} & \multicolumn{3}{c}{Greedy-LS-C}\\
\hfill $k$ & $5$ & $10$ & $50$ & $5$ & $10$ & $50$\\
\midrule
andorra & \numprint{5.2} & \numprint{6.2} & \numprint{5.0} & \numprint{4.2} & \numprint{3.5} & \numprint{26.0}\\
seychelles & \numprint{1.4} & \numprint{5.1} & \numprint{21.0} & \numprint{0.8} & \numprint{3.3} & \numprint{17.6}\\
liechtenstein & \numprint{17.7} & \numprint{14.4} & \numprint{59.7} & \numprint{4.0} & \numprint{15.5} & \numprint{41.3}\\
comores & \numprint{1.7} & \numprint{4.3} & \numprint{7.7} & \numprint{1.2} & \numprint{4.1} & \numprint{19.2}\\
faroe-islands & \numprint{20.2} & \numprint{77.0} & \numprint{254.1} & \numprint{25.7} & \numprint{44.8} & \numprint{189.4}\\
opsahl-openflights & \textless\numprint{0.1} & \numprint{0.1} & \numprint{1.0} & \textless\numprint{0.1} & \numprint{0.1} & \numprint{0.6}\\
tntp-ChicagoRegional & \numprint{45.3} & \numprint{151.4} & \numprint{0.3} & \numprint{32.5} & \numprint{68.2} & \numprint{304.3}\\
\bottomrule
\end{tabular}

%% file: tables/time-c-directed-weighted-high-diameter.tex
\begin{tabular}{lrrrrrr}
\toprule
Graph & \multicolumn{3}{c}{GS-LS-C} & \multicolumn{3}{c}{Greedy-LS-C}\\
\hfill $k$ & $5$ & $10$ & $50$ & $5$ & $10$ & $50$\\
\midrule
andorra & \numprint{5.9} & \numprint{16.0} & \numprint{129.3} & \numprint{6.4} & \numprint{5.8} & \numprint{52.8}\\
seychelles & \numprint{2.1} & \numprint{2.8} & \numprint{59.2} & \numprint{2.3} & \numprint{2.7} & \numprint{29.6}\\
liechtenstein & \numprint{3.7} & \numprint{16.9} & \numprint{227.5} & \numprint{3.8} & \numprint{22.2} & \numprint{167.9}\\
comores & \numprint{1.9} & \numprint{7.0} & \numprint{90.0} & \numprint{2.2} & \numprint{10.7} & \numprint{28.1}\\
faroe-islands & \numprint{16.2} & \numprint{148.1} & \numprint{696.2} & \numprint{15.3} & \numprint{27.2} & \numprint{98.5}\\
\bottomrule
\end{tabular}